\documentclass[showpacs,amsmath,amssymb,twocolumn,aps,pra,superscriptaddress,notitlepage,floatfix,10pt]{revtex4-2}

\usepackage{stylesetting}
\usepackage{orcidlink}
\usepackage[normalem]{ulem}
\usepackage{comment}
\newcommand{\FUBerlin}{Freie Universit\"at Berlin, 14195 Berlin, Germany}
\newcommand{\HZBerlin}{Helmholtz-Zentrum Berlin f\"ur Materialien und Energie, 14109 Berlin, Germany}

\newcommand{\FudanEMW}{Key Laboratory for Information Science of Electromagnetic Waves (Ministry of Education), Fudan University, Shanghai 200433, China}
\newcommand{\PKUCs}{Center on Frontiers of Computing Studies, School of Computer Science, Peking University, Beijing 100871, China}

\newcommand{\FudanSurfacePhysics}{State Key Laboratory of Surface Physics, Key Laboratory of Micro and Nano Photonic Structures (MOE), and Department of Physics, Fudan University, Shanghai 200433, China}

\newcommand{\ShanghaiQiZhiInstitute}{Shanghai Qi Zhi Institute, AI Tower, Xuhui District, Shanghai 200232, China}

\newcommand{\FudanNanoelectronicsQC}{Institute of Nanoelectronics and Quantum Computing, Fudan University, Shanghai 200433, China}

\newcommand{\ShanghaiAILab}{Shanghai Artificial Intelligence Laboratory, Shanghai 200232, China}

\newcommand{\ShanghaiRCQS}{Shanghai Research Center for Quantum Sciences, Shanghai 201315, China}

\makeatletter
\newcommand{\appendixtableofcontents}{%
  \section*{Contents}%
  \@starttoc{atoc}%
}
\newcommand{\startappendixtoc}{%
  \let\originaladdcontentsline\addcontentsline
  \renewcommand{\addcontentsline}[3]{%
    \def\firstarg{##1}%
    \def\tocname{toc}%
    \ifx\firstarg\tocname
      \originaladdcontentsline{atoc}{##2}{##3}%
    \else
      \originaladdcontentsline{##1}{##2}{##3}%
    \fi
  }%
}
\makeatother

\begin{document}

\title{Constant-depth global shadow estimation}

\author{Qingyue Zhang}
\thanks{These authors contributed equally to this work.}
\affiliation{\FudanEMW}

\author{Zhou You}
\thanks{These authors contributed equally to this work.}
\affiliation{\PKUCs}
\affiliation{\FudanEMW}

\author{Dayue Qin}
\thanks{These authors contributed equally to this work.}
\affiliation{\FudanEMW}

\author{Xiaopeng Li}
\email{xiaopeng\_li@fudan.edu.cn}
\affiliation{\FudanSurfacePhysics}
\affiliation{\ShanghaiQiZhiInstitute}
\affiliation{\FudanNanoelectronicsQC}
\affiliation{\ShanghaiAILab}
\affiliation{\ShanghaiRCQS}

\author{Jens Eisert}
\email{jense@zedat.fu-berlin.de}
\affiliation{\FUBerlin}
\affiliation{\HZBerlin}

\author{You Zhou}
\email{you\_zhou@fudan.edu.cn}
\affiliation{\FudanEMW}

\date{September 10, 2026}

\begin{abstract}
Reliable and scalable readout strategies are essential for quantum technologies. As quantum processors grow, extracting useful information must remain feasible without measurement circuits becoming a dominant bottleneck. Randomized measurements and classical shadows provide a powerful route, but global estimation is conventionally associated with highly random ensembles that require increasing circuit depth and hence substantial experimental overhead. In this work, we show that substantially less randomness suffices when the readout is meaningfully adapted to the quantities being estimated. We introduce shallow phase shadows, based on a sparse Clifford-IQP ensemble, and prove efficient global estimation of stabilizer-state fidelities despite the ensemble not forming an approximate relative-error design. On all-to-all architectures, the protocol admits a constant-depth implementation using mid-circuit measurements and classical feedforward, or logarithmic depth without auxiliary systems. The protocol requires only controlled-phase entangling gates and offers a tunable trade-off between circuit resources and estimation accuracy, making it particularly amenable to experimentally relevant architectures with long-range connectivity. Our results show that scalable quantum readout need not reproduce generic randomness: task-adapted randomization can enable substantially shallower global characterization protocols.

\end{abstract}

\maketitle

\section{Introduction}

Stabilizer states play a pivotal role in quantum information science. Perhaps their most prominent application is in quantum error correction and fault tolerance~\cite{campbell2017roads,terhal2015quantum,QECBasic,eisert2025mind}, where the codewords of most known quantum error-correcting codes are stabilizer states~\cite{gottesman1997stabilizer}.
At the same time, graph states provide versatile theoretical laboratories with a wealth of applications~\cite{Graphs,GraphsLong,briegel2009measurement}. Their most prominent instances include cluster states, which form the basis of schemes for measurement-based quantum computation~\cite{briegel2009measurement}. Stabilizer states therefore take center stage in quantum technologies, not only in quantum computing but also in quantum communication, for instance, in the generation and characterization of large-scale entanglement~\cite{guhne2009entanglement,friis2018observation} and in quantum repeaters~\cite{azuma2015all,azuma2023quantum}.
Indeed, given their central role, it is not surprising that the development of new applications and implementation schemes has been accompanied by a growing toolbox for benchmarking and certifying stabilizer states. This includes schemes for state certification~\cite{eisert2020quantum}, entanglement characterization~\cite{cao2023generation,jiang2025generation,graham2022multi}, benchmarking of entire gate sets and other aspects of quantum devices~\cite{arute2019quantum,cao2023generation,liu2025certified,onorati2024noise,RandomSequences}, and validation of quantum error-correcting codes~\cite{wagner2023learning,chen2025quantum,lee2026efficient}. This is part of a wider concerted research effort in benchmarking \cite{BenchmarkingReview,Toolbox},
state certification \cite{PRXQuantum.2.010201},
and quantum learning theory \cite{LearningDeWolf}.

A range of techniques has been developed for this purpose, including direct fidelity estimation~\cite{flammia2011direct,dasilva2011practical} and stabilizer-based verification protocols~\cite{pallister2018optimal,zhu2019efficient,eisert2020quantum}, which leverage prior knowledge of the target state to enable efficient characterization. In contrast, randomized measurement schemes~\cite{elben2023randomized,Efficient,cieslinski2024analysing}, exemplified by classical shadow estimation~\cite{aaronson2019shadow,huang2020predicting}, require no prior knowledge and can estimate exponentially many properties simultaneously from a single dataset. {However, global shadow estimation also requires tractable classical post-processing. This requirement motivates our focus on stabilizer-state fidelities, which form a practically important class of global observables admitting efficient classical computation.
} For global Clifford measurements targeting long-range or highly entangled stabilizer states, the standard shadow protocol requires circuits of linear depth to achieve sufficient randomness (say, a 3-design ensemble), leading to substantial experimental overhead. Shallow-shadow approaches~\cite{bertoni2024shallow,hu2025demonstration} employ logarithmic-depth random Clifford circuits, though their benefits are largely restricted to structured observables such as low-bond-dimension matrix-product operators.

Building on these developments, recent theoretical advances, notably the \textit{gluing} framework~\cite{schuster2024random,cui2025unitary}, have further shown that, on all-to-all architectures with $n$ qubits, efficient shadow estimation can be achieved using circuits of depth $O(\log\log n)$ through approximate relative-error designs.
Nevertheless, reducing circuit depth remains highly desirable for scalable implementations on long-range quantum architectures~\cite{eisert2025mind}.
This naturally raises the question: \textit{can efficient global shadow estimation be achieved in constant depth, at least for important classes of highly entangled states such as stabilizer states?}

\begin{figure*}[htbp]
    \centering
    \includegraphics[width=\linewidth]{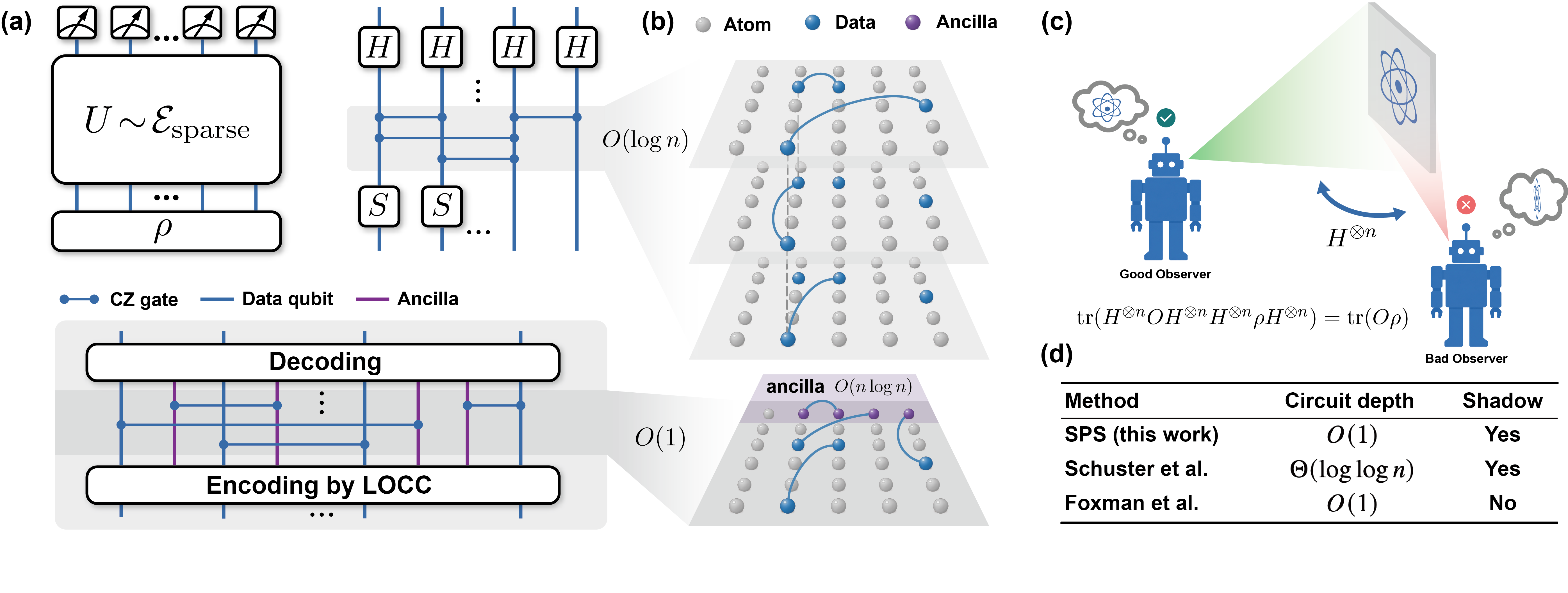}
\caption{
Overview of shallow phase shadows.
(a) Global shadow estimation of an unknown state $\rho$ using sparse Clifford-IQP $S$-$CZ$-$H$ circuits.
(b) The sampled sparse circuit can be implemented either in constant depth on all-to-all architectures using LOCC-assisted encoding and decoding with $O(n\log n)$ auxiliary systems or in logarithmic depth without auxiliary systems.
(c) The second-look principle measures two complementary global observables $O$ and $H^{\otimes n} O H^{\otimes n}$. We then select the lower-bias track in post-processing, enabling accurate shadow estimation.
(d) SPS achieves constant-depth shadow estimation beyond the light-cone restriction of the gluing framework~\cite{schuster2024random}, while existing constant-depth randomization results provide weaker guarantees~\cite{foxman2025random}, which are insufficient for full shadow guarantees over stabilizer observables.
}
    \label{fig:const_depth_impl}
\end{figure*}

Here we answer this question in the affirmative. We introduce \emph{shallow phase shadows} (SPS), a global shadow-estimation protocol that provably evades the $\Omega(\log\log n)$ depth lower bound~\cite{schuster2024random} and achieves constant depth on all-to-all architectures. SPS relies on two key ingredients. First, the commuting diagonal structure of the $S$-$CZ$-$H$ gate ensemble enables a space-time trade-off: using mid-circuit measurements and classical feedback~\cite{buhrman2024state,zi2025constant,vazquez2024scaling}, the diagonal layer can be implemented in constant depth on all-to-all architectures with the same asymptotic auxiliary-system overhead as Ref.~\cite{schuster2024random}. Second, we develop a \emph{second-look} principle that enables efficient global shadow estimation even though the underlying circuit ensemble does not form an approximate design. Conceptually, this overcomes an apparent tension in shadow estimation. The flexibility of classical shadows is often associated with sufficiently generic random measurements, while measurement schemes tailored to particular observables appear to sacrifice the ``measure first, ask questions later'' paradigm. SPS shows that this dichotomy is not necessary: the randomization can be adapted to a structured class of estimation tasks while the particular observable to be estimated is still chosen only in post-processing. For stabilizer-state fidelities, substantially less randomness can therefore retain the essential flexibility of shadow estimation while dramatically reducing the required circuit depth. Furthermore, SPS requires only $CZ$ entangling gates, which are native to leading long-range platforms such as trapped ions~\cite{decross2025computational,figgatt2019parallel,liu2025certified} and neutral atoms~\cite{bluvstein2026faulttolerant,evered2023high,bluvstein2024logical}. Taken together, our results show that scalable global shadow estimation need not reproduce generic randomness: task-adapted randomization can preserve the defining flexibility of classical shadows while enabling substantially shallower, hardware-friendly characterization protocols.

\section{Overview of main results}

Our primary contribution is the \emph{shallow phase shadow} (SPS) protocol, which enables shadow estimation using \emph{constant-depth} circuits. In contrast to the \textit{gluing} framework~\cite{schuster2024random}, which achieves depth $O(\log\log n)$ through patch constructions, SPS achieves constant depth on all-to-all architectures by trading circuit depth for auxiliary quantum systems. Our main result can be summarized as follows.

\begin{theorem}[Shallow phase shadow in constant depth, informal]\label{th:constdepth}
Run the shallow phase shadow protocol on the all-to-all architecture, guaranteeing that for every stabilizer-state observable $O$, the resulting estimator has bias at most a fixed value  $\epsilon\in(0,1)$. Then, with overwhelmingly high probability, the following quantum circuit implementation is sufficient:
\begin{itemize}
\item  $O(1)$-depth, using
    $O\!\left(n\log\frac{n}{\epsilon}\right)$ auxiliary qubits,
\item $O\!\left(\log\frac{n}{\epsilon}\right)$-depth, without using auxiliary systems.
\end{itemize}
\end{theorem}

The formal statement and proof of \cref{th:constdepth} are provided in \cref{Ap:proofmain}. The argument combines the performance guarantee of SPS from \cref{th:PshadowMain} with the circuit-resource estimates of \cref{prop:resource_limits}. Specifically, one first chooses the sparsity parameter $\gamma$ of the quantum circuit such that the SPS estimator achieves bias at most $\epsilon$, and then translates this choice into explicit depth or auxiliary-system requirements.

\cref{th:constdepth} shows that the SPS protocol admits a constant-depth implementation using $O(n\log\frac{n}{\epsilon})$ auxiliary systems. Notably, this auxiliary-system overhead matches the scaling of state-of-the-art shallow-shadow protocols~\cite{schuster2024random}.
This is particularly well suited to platforms such as neutral atoms~\cite{bluvstein2026faulttolerant,evered2023high,bluvstein2024logical}, where large and reconfigurable qubit arrays make spatial overhead comparatively affordable and scalability is a central hardware advantage.
Without auxiliary systems, the protocol can still be realized with two-qubit depth $O\!\left(\log\!\frac{n}{\epsilon}\right)$ with high probability.

Our circuit ensemble relies on the commuting diagonal structure of the sparse Clifford-IQP ($S$-$CZ$-$H$) ensemble, based on
\emph{instantaneous quantum polynomial-time} (IQP) circuits.
Unlike general Clifford circuits, this ensemble allows the entire diagonal layer to be executed in a single parallel step~\cite{paletta2024robust}.
A key point of SPS is that efficient global shadow estimation of stabilizer fidelities nevertheless remains possible via a simple two-track measurement strategy, which we refer to as the \emph{second-look} principle.
By contrast, for general ensembles designed to achieve approximate designs~\cite{laracuente2026approximate}, known lower bounds still
enforce at least $\Omega(\log\log n)$ two-qubit depth~\cite{schuster2024random}
on all-to-all architectures with arbitrary auxiliary systems. Constant-time randomization proposals typically provide only weaker guarantees, such as additive errors or measurable errors~\cite{foxman2025random,lee2026shallow}, which are insufficient for shadow estimation. The comparison is summarized in the table in \cref{fig:const_depth_impl}(d), emphasizing that SPS provides constant-depth shadow-estimation guarantees beyond the approximate-design regime.

The sparse Clifford-IQP ensemble behind SPS offers several practical advantages.
First, SPS requires only $CZ$ gates as entangling gates, making it naturally suited to platforms with native long-range interactions, such as neutral atoms and trapped ions. This is particularly advantageous on platforms where entangling gates are comparatively slow, as reducing circuit depth directly mitigates this experimental bottleneck.
Second, the circuit cost can be continuously tuned through the parameter $\gamma$, enabling smooth resource adjustment without changing the circuit architecture. This differs from patch-based gluing constructions~\cite{schuster2024random}, in which the resource cost changes only discretely and the patch size usually must divide $n$.
Finally, SPS delivers a clear quantitative advantage in terms of the actual entangling-gate budget. The numerical simulations in Fig.~\ref{fig:guarentee} show that SPS reaches the same target bias with \emph{several-fold} fewer two-qubit gates than shallow-shadow schemes using existing random circuit ensembles~\cite{schuster2024random,bertoni2024shallow}.

The remainder of this work is organized as follows.
In \cref{sec:framework}, we introduce the sparse Clifford-IQP circuit ensemble underlying SPS and derive an explicit expression for its second-moment operator. We then show how the commuting diagonal structure of this ensemble enables a constant-depth implementation on all-to-all architectures, together with explicit resource estimates both with and without auxiliary systems.

In \cref{sec:shadow}, we construct the shallow phase shadow protocol based on the sparse Clifford-IQP ensemble. We first explain how the deviation of the sparse-ensemble moment from the dense diagonal-design moment leads to a systematic bias under the simple phase-shadow post-processing. We then introduce the second-look principle, which yields a uniform vanishing-bias guarantee for stabilizer-state observables.

In \cref{sec:unbiased}, we provide an exactly unbiased variant of SPS by explicitly inverting the induced measurement channel in the Pauli basis. This removes the systematic bias entirely, at the cost of additional classical post-processing, and complements the lightweight biased estimator developed in \cref{sec:shadow}.

\section{Circuit ensemble and implementation} \label{sec:framework}

A key feature of the shallow phase shadow protocol is the simplicity of the underlying quantum circuits. The required random ensemble consists only of single-qubit phase and Hadamard gates together with a sparse set of $CZ$ gates. Despite this simple structure, the ensemble provides sufficient randomness for the global estimation tasks considered here and, crucially, admits a particularly shallow and experimentally feasible implementation. In this section, we first introduce the sparse Clifford-IQP ensemble and characterize its moment operators, which will later allow us to bound the bias and variance of the shadow estimator. We then exploit the commuting structure of the diagonal gates to show that every sampled circuit can, with overwhelmingly high probability, be realized either in constant depth on all-to-all architectures using a moderate number of auxiliary systems, or in logarithmic depth without additional quantum systems. Finally, we discuss how these resource requirements translate into experimental implementations and identify long-range quantum platforms for which the resulting space-time trade-off is particularly natural.

\subsection{The sparse Clifford-IQP ensemble}

We consider an ensemble of random $n$-qubit unitaries of the $S$-$CZ$-$H$ form, where
the diagonal gates ${CZ}=\operatorname{diag}(1,1,1,-1)$,
${S}=\operatorname{diag}(1,i)$,
and ${H}$ denotes the Hadamard gate. The diagonal layer is randomized, while the final Hadamard layer is fixed, so that
\begin{equation}
U = H^{\otimes n} U_S U_C \in \mathcal{E}_{\mathrm{sparse}}^{(\gamma)},
\end{equation}
with
\begin{equation}\label{Eq:diagonalcircuits}
U_S=\prod_{k\in[n]} S_k^{A_{k,k}},
\qquad
U_C=\prod_{0\le i<j\le n-1} CZ_{i,j}^{A_{i,j}},
\end{equation}
where $A$ is a random symmetric binary matrix. The diagonal entries $A_{i,i}$ are independently and identically distributed (i.i.d.) as Bernoulli$(1/2)$, while the off-diagonal entries $A_{i,j}$ ($i<j$) are i.i.d. as Bernoulli$(p)$ with
\begin{equation}
p=\frac{\gamma\ln n}{n},
\end{equation}
where $\gamma$ is a constant sparsity parameter.

This construction is related to the sparse IQP circuits introduced in Ref.~\cite{bremner2017achieving}, which take the form $C=H^{\otimes n}DH^{\otimes n}$ with a diagonal unitary $D$~\cite{shepherd2009temporally,bremner2011classical,bremner2016average}. In contrast, our ensemble employs only a single Hadamard layer and restricts the diagonal gates to the Clifford set generated by $S$ and $CZ$, excluding non-Clifford phase gates like controlled-$S$ and $T$ gates. We therefore refer to $\mathcal{E}_{\mathrm{sparse}}^{(\gamma)}$ as the \emph{sparse Clifford-IQP ensemble} hereafter.

To connect the shadow statistics with properties of the circuit ensemble, we introduce an $m$-th moment operator. For each sampled unitary $U$ and the computational basis indices $\mb{b}\in\{0,1\}^n$, let $\Phi_{U,\mb{b}} := U^\dagger \ket{\mb{b}}\!\bra{\mb{b}}\,U$.
Then the $m$-th moment function of a unitary ensemble $\mc{E}$ is defined as
\begin{equation}
\mb{M}_{\mc{E}}^{(m)} := D^{-1}\, \mathbb{E}_{U\sim \mc{E}}\sum_{\mb{b}} \Phi_{U,\mb{b}}^{\otimes m},
\end{equation}
with $D=2^n$. For the sparse Clifford-IQP ensemble $\mc{E}_{\text{sparse}}^{(\gamma)}$ here, the symmetry of the Hadamard and diagonal phase layers allows the sum over computational basis outcomes to reduce to a fixed string (e.g., $\mb{b}=\mb{0}$), yielding the equivalent representation
\begin{equation}\label{eq:Moment2def}
\mb{M}_{\mc{E}_{\text{sparse}}^{(\gamma)}}^{(m)}
= \mathbb{E}_{U\sim\mc{E}_{\text{sparse}}^{(\gamma)}}
\bigl(U^{\dagger}\ket{\mb{0}}\!\bra{\mb{0}}U\bigr)^{\otimes m}.
\end{equation}

Before providing the analytical result of \cref{eq:Moment2def} with  $m=2$ in \cref{thm: sparse 2-moment}, we give some definitions. Denote $\id_4$ and $\mathbb{S}_2$ as the two-copy identity and swap operators, respectively, and
$\Delta_2:=\ketbra{0,0}{0,0}+\ketbra{1,1}{1,1}$
as the two-copy parity projector, which are all supported on
a qubit-pair. Here, the second moment in~\cref{eq:Moment2def} admits a decomposition into tensor products of these three elementary building blocks. In Ref.~\cite{zhang2025robust}, we have shown a special case where each CZ gate is sampled with  probability $1/2$, corresponding to $\gamma^*=\frac{n}{2\ln n}$, the second moment shrinks to the compact form
\begin{equation}\label{eq:m2RPS}
    \mb{M}_{\mc{E}_{\text{sparse}}^{(\gamma^*)}}^{(2)}=D^{-2}(\id_4^{\otimes n}+\mbb{S}_2^{\otimes n}-\Delta_2^{\otimes n}).
\end{equation}

Despite this simplified expression, obtaining such a second moment requires a circuit depth of order $O(n)$, which is comparatively large. For the general constant $\gamma$ considered here, however, the second moment has a richer and more complex expansion. To state the result, let $I_i,I_j,I_k\subseteq [n]$ be mutually disjoint index sets satisfying $I_i+I_j+I_k=[n]$, where "$+$" denotes disjoint union. One writes $i:=|I_i|$, $j:=|I_j|$, and $k:=|I_k|$ for their cardinalities.  We use the shorthand
\begin{equation}
\Delta_2^{\otimes I_i}\otimes (\id_4-\Delta_2)^{\otimes I_j}\otimes (\mathbb{S}_2-\Delta_2)^{\otimes I_k}
\end{equation}
for the $n$-qubit(-pair) tensor product in which all the qubit-pairs with indices  in $I_i$ take $\Delta_2$,
similar for $I_j$ and $I_k$.

\begin{lemma}[Second moment of the sparse ensemble]\label{thm: sparse 2-moment}
The second moment function of the sparse Clifford-IQP ensemble $\mc{E}_{\text{sparse}}^{(\gamma)}$ defined in \cref{Eq:diagonalcircuits} is given by
\begin{equation}
\begin{split}
\mathbf{M}_{\mc{E}_{\text{sparse}}^{(\gamma)}}^{(2)}
&=D^{-2} \sum_{I_i,I_j,I_k}
\Bigl(1-\frac{2\gamma \ln n}{ n}\Bigr)^{jk}\\
&\quad{}\times\Delta_2^{I_i}\otimes (\id_4-\Delta_2)^{I_j} \otimes (\mathbb{S}_2-\Delta_2)^{I_k},
\end{split}
\end{equation}
where $I_i\in\mc{I}_i$, $I_j\in\mc{I}_j$, and $I_k\in\mc{I}_k$ are mutually disjoint index sets with $I_i+I_j+I_k=[n]$.
\end{lemma}
\Cref{thm: sparse 2-moment} expresses the second moment as a sum of tensor products of three elementary operator components — a trinomial structure — which is essential for proving the bias result in \cref{th:PshadowMain}. Because the ensemble is permutation invariant, the coefficient in this decomposition depends only on the cardinalities $j$ and $k$, not on the specific locations of the operators. The proof of \cref{thm: sparse 2-moment} is given in \cref{Ap:moment2}. The corresponding third-moment formula, used to prove the variance bound in \cref{th:PshadowMain}, is deferred to \cref{Ap:moment3}.

As a consistency check, for the previous special case $p=1/2$, or equivalently
$\gamma=n/(2\ln n)$, the coefficient satisfies
\begin{equation}\left(1-\frac{2\gamma\ln n}{n}\right)^{jk}=0
\end{equation}
whenever $j,k>0$. Therefore, only the terms with $j=0$ or $k=0$ survive, and \cref{thm: sparse 2-moment} reduces exactly to the compact second-moment expression in \cref{eq:m2RPS} by summing the binomial terms.

A few studies have characterized the distance between the moments of diagonal circuit ensembles and that of Haar random ensemble~\cite{nakata2014generating,ji2018pseudorandom}. However, these works typically focus on additive or measurable errors~\cite{cui2025unitary}. Such guarantees are often insufficient for shadow estimation, as the additive inaccuracies can be exponentially amplified by the inversion of the measurement channel~\cite{huang2020predicting}, leading to a large estimation bias.
The relative error, by contrast, is a stronger measure, as it requires the $k$-th moment $\mb{M}_{\mc{E}}^{(k)}$ of the ensemble to multiplicatively approximate the Haar measure’s moment $\mb{M}_{\text{Haar}}^{(k)}$
\begin{equation}\label{eq:RelativeDef}
(1-\epsilon) \mb{M}_{\text{Haar}}^{(k)} \preceq \mb{M}_{\mc{E}}^{(k)} \preceq (1+\epsilon) \mb{M}_{\text{Haar}}^{(k)}.
\end{equation}

This condition provides multiplicative control over all second-moment directions, and is therefore often regarded as a standard route to rigorous shadow-estimation guarantees~\cite{choi2023preparing}. However, SPS follows a different route. As shown later in \cref{lem:NotRelDesign}, the sparse Clifford-IQP ensemble considered here does not form a good relative-error state $2$-design. The key point is that stabilizer-shadow estimation does not require such uniform control over all directions. Instead, SPS uses the second-look principle (introduced in \cref{subsec:HadamardTrick}) to enable accurate estimation even beyond the relative-error design condition.

\subsection{Constant-depth synthesis and cost analysis}\label{subsec:circuit}

In this subsection, we give an explicit implementation of the sampled unitary
$U=H^{\otimes n}U_SU_C$ from \cref{Eq:diagonalcircuits} on an all-to-all architecture using mid-circuit measurement and classical feedback~\cite{buhrman2024state,alam2024learning,yan2025variational,liu2026characterization}. We also quantify the corresponding resource cost in \cref{prop:resource_limits}. The key point is that $U_SU_C$ is a product of commuting diagonal gates. For every qubit $q$, let $G=(V,E)$ be the interaction graph induced by the sampled $CZ$ gates, and set the rail length $r_q := \deg_G(q)$.

As shown in \cref{fig:const_depth_impl}, each qubit is first encoded into a GHZ-like rail of length $r_q$, which converts temporal circuit complexity into spatial overhead. The implementation proceeds in the following three steps.

1. Encoding (fan-out). For every qubit $q$, we encode it into a GHZ-like rail
$\{\bar q^{(0)},\ldots,\bar q^{(r_q-1)}\}$ so that
$\alpha|0\rangle_q+\beta|1\rangle_q \mapsto \alpha|0,\cdots, 0\rangle_{\bar q}
+\beta|1,\cdots ,1\rangle_{\bar q}$.

2. Gate assignment. We assign the local gate $S_q$ to the
reserved slot $\bar q^{(0)}$. For each edge $(i,j)\in E$, we choose a unique slot on both rails, and apply a single physical
$CZ_{\bar i^{(\ell_{i,j})},\,\bar j^{(\ell_{i,j})}}$. By construction, each rail qubit $\bar q^{(\ell)}$ participates in at most one two-qubit gate, so all
$CZ$ gates are executed simultaneously in one layer.

3. {Decoding (fan-in).} We reverse the fan-out to disentangle all auxiliary systems and return the logical state to the
original data qubits. Finally, we apply the fixed $H^{\otimes n}$ layer and
measure in the computational basis.

The circuit cost of this implementation is as follows. The fan-out and inverse fan-in are both constant-depth procedures when assisted by mid-circuit measurements and classical feedback~\cite{buhrman2024state,zi2025constant}. Using the construction of Ref.~\cite{zi2025constant}, a rail of length $r_q$ requires an additional $r_q/c$ auxiliary qubits, so the encoded rail uses $(1+c^{-1})r_q$ qubits in total. For convenience, we set $c=2$, in which case the encoding and decoding steps require about 5 and 3 circuit layers, respectively. Correspondingly, the total physical
qubit count becomes $n_{\mathrm{tot}}=\sum_{q\in V}(1+c^{-1})r_q.$
We denote $|E|$ as the number of edges in the random interaction graph
$G(n,p)$. Since the total number of rail qubits is
$n_{\mathrm{rail}}=\sum_{q\in V} r_q = 2|E|$, we obtain that the number of auxiliary systems $n_{\mathrm{anc}}$
required by the fan-out/fan-in procedure is bounded by the total physical
qubit count $n_{\mathrm{tot}}$
given by
\begin{equation}\label{eq:ancilla}
n_{\mathrm{anc}}< n_{\mathrm{tot}}
=
\sum_{q\in V} (1+c^{-1})r_q
=
2(1+c^{-1})|E|.
\end{equation}
In particular, for the choice $c=2$, this gives $n_{\mathrm{tot}}=3|E|$.

We now quantify the resource requirements for this synthesis.

\begin{prop}[Resource estimation of quantum circuit realization, informal]\label{prop:resource_limits}
For the sparse unitary ensemble $\mc{E}_{\text{sparse}}^{(\gamma)}$ defined around \cref{Eq:diagonalcircuits}, the quantum circuit realization satisfies the following resource bounds %
\begin{itemize}
    \item[(i)] $O(1)$-depth, using no more than $3\gamma n \ln n$ auxiliary systems.
    \item[(ii)] $2\gamma \ln n$-depth, without using auxiliary systems,
\end{itemize}
with overwhelmingly high probability.
\end{prop}

The formal statement and proof of \cref{prop:resource_limits} is left to~\cref{Ap:ProofProp1}. It shows that the stated depth and auxiliary system bounds hold with overwhelmingly high probability over the random choice of the sparse Clifford-IQP ensemble.
Compared with generic Clifford circuits, the diagonal structure of sparse IQP circuits enables an explicit space-time trade-off, which is the key ingredient allowing a constant-depth realization.
{This naturally raises the question of whether the same strategy extends to other $CZ$-based architectures, including more general IQP circuits \cite{cao2026measurement,SupremacyReview},
other quantum advantage schemes \cite{Supremacy} and \emph{measurement-based quantum computation} (MBQC) \cite{raussendorf2001one,Raussendorf2003}.
Whether these architectures can likewise help remains an interesting open question.}

\subsection{Experimental implementation}
\label{sec:exp}

The simplicity and low depth of SPS make the protocol particularly attractive for experimental implementation on platforms with long-range connectivity, notably reconfigurable neutral-atom arrays~\cite{bluvstein2026faulttolerant,evered2023high,bluvstein2024logical} and trapped-ion processors~\cite{decross2025computational,figgatt2019parallel,liu2025certified,Helios}. The required circuit consists only of single-qubit $S$ and $H$ gates, sparse $CZ$ entangling gates, measurements, and classical feedforward. In the constant-depth realization, the $O(n\log n)$ auxiliary-system overhead can be understood physically as a spatial resource: auxiliary qubits provide the rails required to distribute the interactions, corresponding, for instance, to additional atom rearrangements or shuttling operations on reconfigurable architectures. Importantly, the quantum depth remains constant, while the expected entangling-gate budget is only $O(n\log n)$. The resulting experimental resource requirements are therefore comparatively  low and can be continuously tuned through the sparsity parameter $\gamma$. This is particularly relevant in the presence of physical errors, where keeping quantum circuits shallow (and, in fact, close to constant in depth) is itself an important route to limiting noise accumulation~\cite{Nonunital}. These features make SPS a natural candidate for implementation on present and emerging long-range quantum processors.

\section{The shallow phase shadow protocol}\label{sec:shadow}

In this section, we detail the construction of the shallow phase shadow protocol utilizing the
sparse Clifford-IQP ensemble discussed above. The core challenge in utilizing constant-depth circuits for shadow estimation is that the resulting unitary ensemble does not form an approximate state 2-design, leading to large estimation bias. We address this by strictly decomposing the estimation task into several parts and rigorously bounding the residual error.

\subsection{Phase shadow with sparse Clifford-IQP ensemble}

Given the sparse Clifford-IQP ensemble $\mc E_{\text{sparse}}^{(\gamma)}$, we now introduce the shallow phase shadow protocol.
Shadow estimation is a quantum-classical hybrid process~\cite{hadfield2022measurements,zhou2024hybrid,zhou2023performance,hu2023classical,liu2026auxiliary,farias2024robust,chen2021robust,li2025nearly,huang2021efficient,nguyen2022optimizing,you2025circuit}. On the quantum side, a random unitary $U\in\mc{E}$ is applied to an unknown quantum state $\rho\mapsto U\rho U^{\dagger}$ and one measures it in the computational basis to get the
classical outcome word $\mb{b}$,
with conditional probability $\Pr(\mb{b}|U)=\langle \mb{b}| U\rho U^{\dag} |\mb{b}\rangle=\tr(\rho \Phi_{U,\mb{b}})$. On the classical side, one `prepares' a quantum snapshot $\Phi_{U,\mb{b}}$ on the classical computer.

Taking the expectation on both $U$ and $\mb{b}$, the overall process is given by the channel
\begin{equation}\label{Eq:idealchannel}
\begin{aligned}
\mc{M}_{\mc{E}}(\rho):
    &=
    \mbb{E}_{U\sim\mc{E},\mb{b}}\ \Phi_{U,\mb{b}}=\sum_{\mb{b}}\mbb{E}_{U\sim\mc{E}} \tr_{1}(\rho\otimes \id \text{ }\Phi_{U,\mb{b}}^{\otimes 2})\\
    &=D \tr_1\bigl[(\rho\otimes \id)\ \mb{M}_{\mc{E}}^{(2)}\bigl],
\end{aligned}
\end{equation}
where $\tr_{1}$ denotes the (partial) trace over
the first copy, and in the second line we relate it to the predefined moment function.

As a special case, when $p=\frac12$ (equivalently $\gamma^*=n/(2\ln n)$), its second moment reduces to that of a diagonal state-2 design~\cite{nakata2014generating,zhang2025robust}, with the moment function given in \cref{eq:m2RPS}. In this case, the single-shot estimator or the post-processing strategy shows
\begin{equation}\label{Eq:shadow2}
\widehat{\rho_f}:=D\Phi_{U,\bb}-\id,
\end{equation}
which is unbiased for the off-diagonal part $\rho_f\coloneq\rho-\rho_d$ \cite{zhang2025robust}. And the diagonal part is $\rho_d=\tr_1\bigl[\Delta_D (\rho \otimes \id_D)\bigr]=\sum_{\mb{b}}\ketbra{\mb{b}}{\mb{b}}\rho\ketbra{\mb{b}}{\mb{b}}$ for $\Delta_D \coloneq \Delta_2^{\otimes n}$.
It is important to note that
while this estimator targets the off-diagonal $\rho_f$, the diagonal components $\rho_d$ can be efficiently estimated separately using simple computational basis measurements (i.e., setting $U=\id$). Combining these two procedures allows for full state reconstruction \cite{park2023resource,zhang2025robust}.

\subsection{The two-track measurement for precise estimation}
\label{subsec:HadamardTrick}

The shallow phase shadow protocol defines a hardware-friendly shadow estimator based on sparse Clifford-IQP circuits.
However, the absence of arbitrarily accurate relative-error state $2$-design behavior can lead to a non-vanishing systematic bias for certain stabilizer-state observables.
In this subsection, we illustrate this obstruction through explicit examples, and then show that, despite this limitation, one can still achieve efficient shadow estimation by adapting the protocol to the structure of the target observables.

We now apply the post-processing rule in \cref{Eq:shadow2} to the SPS protocol. However, the second-order moment of SPS in \cref{thm: sparse 2-moment} differs from the ideal moment in \cref{eq:m2RPS}. This difference can be viewed as an effective channel noise:
The simple post-processing rule in \cref{Eq:shadow2} no longer gives an exact inversion and therefore introduces a systematic bias.

For a stabilizer-state observable $O$, we quantify this bias by
\begin{equation}\label{Eq:BiasDef}
\mathrm{Bias}(\widehat{o_f})
:=
\Bigl|
\mathbb{E}_{U,\mb b}\tr\!\bigl(O\,\widehat{\rho_f}\bigr)
-
\tr\!\bigl(O\,\rho_f\bigr)
\Bigr|.
\end{equation}
Below, we illustrate that this bias can remain large for some stabilizer observables.
\begin{figure}[htbp]
\centering \includegraphics[width=\linewidth]{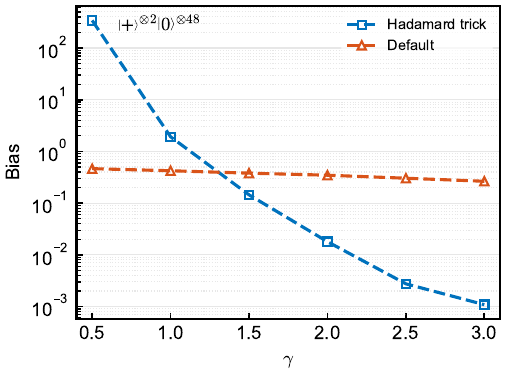} %

    \caption{Scaling of the estimation bias with respect to the parameter $\gamma$ for the worst-case state vector $\left| + \right\rangle^{\otimes 2}\left| 0 \right\rangle^{\otimes n-2}$ with $n=50$. The performance of the Hadamard-conjugated task ($\rho',O'$, blue squares) is compared with the default task (orange triangles).}
    \label{fig:hadamard_trick}
\end{figure}
A simple stabilizer-fidelity instance already exhibits that the bias is potentially large. Let $\rho = O = \ket{\psi}\bra{\psi}$ with $\ket{\psi}=\ket{+}^{\otimes 2}\otimes\ket{0}^{\otimes(n-2)}$, and

\begin{equation}\label{Eq:exHadatrick_1}
\mathrm{Bias}(\widehat{o_f}) \ge \frac{1}{2}\Bigl(1-\frac{2\gamma\ln n}{n}\Bigr),
\end{equation}
with the proof deferred to \cref{Ap:ProofOfHt1}.
In particular, for fixed $\gamma$, this lower bound remains bounded away from zero as $n\to\infty$. This obstruction is intrinsic to the sparse Clifford-IQP ensemble, as formalized below.

\begin{lemma}[Failure of relative-error state $2$-design]
\label{lem:NotRelDesign}

The sparse Clifford-IQP ensemble $\mc E_{\mathrm{sparse}}^{(\gamma)}$ defined in \cref{Eq:diagonalcircuits}
\emph{cannot} form an $\epsilon$-approximate state $2$-design in relative error
for any $\epsilon < \tfrac18 - o(1)$.
\end{lemma}

In other words, the ensemble remains a constant relative error from the exact state $2$-design.
The proof of \cref{lem:NotRelDesign} is left to \cref{Ap:NotRelDesign}.

To overcome this obstruction, we consider the Hadamard-conjugated task
$(\rho',O')=\bigl(H^{\otimes n}\rho H^{\otimes n},\,H^{\otimes n}OH^{\otimes n}\bigr)$,
and apply the same shadow estimation procedure to $(\rho',O')$. Since
$\Tr(O\rho)=\Tr(O'\rho')$,
conjugating both the state and the observable by the same unitary leaves the quantity of interest invariant.
In \cref{fig:hadamard_trick}, we numerically compare the estimation bias obtained from the original task $(\rho,O)$ and from its Hadamard-conjugated counterpart. This numerical result indicates that an appropriate global basis choice can dramatically improve the estimation
accuracy.

Motivated by the above example, we introduce a general \emph{second-look principle}. A complete description of the resulting SPS protocol is given in \cref{algo:PhaseShadow}. Here, we collect snapshots in two parallel tracks, one for $(\rho,O)$ and one for $(\rho',O')$, and we defer the choice of which track to use for estimating a given stabilizer-state observable until post-processing.
For stabilizer-state observables, this choice can be made efficiently from the stabilizer tableau~\cite{aaronson2004improved}. Writing
\begin{equation}
T(V)=\begin{bmatrix}A&B\\ C&D\end{bmatrix},
\end{equation}
for the tableau of
$O=V\ket{\mathbf 0}\!\bra{\mathbf 0}V^\dagger$, \cref{Ap:weight} shows that the estimator bias is governed by
$\zeta=\rank(C)$ in the original track (and
$\zeta'=\rank(D)$ in the Hadamard-conjugated track). The track with the larger rank parameter has the smaller certified bias. Since both tracks are measured before this observable-dependent choice is made, the protocol preserves the `measure first, ask questions later' philosophy of classical shadows.

Finally, we note that the estimator in~\cref{Eq:shadow2} only targets the off-diagonal component $\rho_f$. Therefore, in \cref{algo:PhaseShadow}, we also estimate the corresponding diagonal components $\rho_d$ (and $\rho'_d$) using computational-basis measurements. These diagonal estimators are unbiased, so we can focus solely
on the bias arising from the off-diagonal contribution.

\begin{algorithm}[H]
\caption{Shallow phase shadow estimation with the second-look principle.}
\label{algo:PhaseShadow}
\KwIn{Snapshots $N_f+N_d$, input state $\rho$, stabilizer-state observables $\{O_t=V_t\ket{\mb{0}}\!\bra{\mb{0}}V_t^\dagger\}_{t=1}^L.$}
\KwOut{Fidelity estimators $\{\widehat{o}^{(t)}:=\tr(\widehat{\rho}\,O_t)\}_{t=1}^L.$}
\For{$s=1,2,\dots,N_f/2$}{
Sample unitary $U^{(s)}\sim\mc{E}_{\mathrm{sparse}}^{(\gamma)}$.\;
Apply $U^{(s)}$ to $\rho$, measure in the $Z$ basis, obtain $\mb{b}_1^{(s)}$ and set $\Phi_1^{(s)}=U^{(s)\dagger}\ket{\mb{b}_1^{(s)}}\!\bra{\mb{b}_1^{(s)}}U^{(s)}$.\;
Apply $U^{(s)}$ to $H^{\otimes n}\rho H^{\otimes n}$, measure in the $Z$ basis, obtain $\mb{b}_2^{(s)}$ and set $\Phi_2^{(s)}=U^{(s)\dagger}\ket{\mb{b}_2^{(s)}}\!\bra{\mb{b}_2^{(s)}}U^{(s)}$.\;
}
\For{$s^\star=1,2,\dots,N_d/2$}{
Prepare $\rho$, measure in the $Z$ basis, obtain $\mb{b}_1^{(s^\star)}$, and set $\Psi_1^{(s^\star)} = \ket{\mb{b}_1^{(s^\star)}}\bra{\mb{b}_1^{(s^\star)}}$.\;
Prepare $H^{\otimes n}\rho H^{\otimes n}$, measure in the $Z$ basis, obtain $\mb{b}_2^{(s^\star)}$, and set $\Psi_2^{(s^\star)} = \ket{\mb{b}_2^{(s^\star)}}\bra{\mb{b}_2^{(s^\star)}}$.\;
}
\For{$t=1,2,\dots,L$}{
Write the tableau of $V_t$ as $\begin{bmatrix}A_t & B_t \\ C_t & D_t\end{bmatrix}$,\;
\eIf{$\mathrm{rank}(C_t)>\mathrm{rank}(D_t)$}{
$\widehat{o_f}^{(t)}\gets \textbf{MedianOfMeans}(\{\Phi_1^{(s)}\}_{s=1}^{N_f/2}).$\;
$\widehat{o_d}^{(t)}\gets \textbf{MedianOfMeans}(\{\Psi_1^{(s^\star)}\}_{s^\star=1}^{N_d/2})$.
}{
$\widehat{o_f}^{(t)}\gets \textbf{MedianOfMeans}(\{\Phi_2^{(s)}\}_{s=1}^{N_f/2}).$\;
$\widehat{o_d}^{(t)}\gets \textbf{MedianOfMeans}(\{\Psi_2^{(s^\star)}\}_{s^\star=1}^{N_d/2})$\;
}
$\widehat{o}^{(t)}\gets \widehat{o_f}^{(t)}+\widehat{o_d}^{(t)}.$\;
}
\end{algorithm}

\subsection{Theoretical guarantees and limitations}

We now quantify the performance of the SPS framework, showing that it achieves a small estimation bias for all stabilizer-state observables. Let $\mathrm{negl}(n)$ denote a term vanishing faster than any inverse polynomial in $n$.

\begin{theorem}[Performance guarantee of shallow phase shadow]\label{th:PshadowMain}
For any stabilizer-state observable $O$, the expected estimation bias using \cref{algo:PhaseShadow} is upper bounded by
\begin{equation}\label{Eq:Bias}
\text{Bias}(\widehat{o_f}) \leq (4n^2+2n)n^{-0.4\gamma} + \mathrm{negl}(n),
\end{equation}
provided $\gamma > 3$. Furthermore, the estimation variance scales sub-exponentially in $n$.
\end{theorem}

This theorem is the central performance guarantee of SPS and also our main contribution. Its key implication is that a constant circuit sparsity parameter $\gamma$ already suffices to make the estimation bias controllable for all stabilizer-state observables. Since the expected two-qubit gate count is $O(n^2 p) = O(\gamma n \ln n)$, this keeps the sampled circuit in the sparse regime required for the log-depth relaxation and constant-depth implementation with a low number of auxiliary systems,  as shown in \cref{subsec:circuit}. More explicitly, \cref{Eq:Bias} implies that achieving bias $\epsilon$ only requires an average gate cost of order $\frac{5}{4} n \ln(6n^2 / \epsilon)$, and hence the constant-depth realization uses $O(n \log(n/\epsilon))$ auxiliary systems,  matching the asymptotic scaling of Ref.~\cite{schuster2024random}.

The complete proof of~\cref{th:PshadowMain} is deferred to~\cref{Ap:ProofOfMainThm}; it is highly nontrivial. Before sketching the proof, we first highlight the key difficulty underlying the bias bound in~\cref{Eq:Bias}. As we have shown that such an ensemble $\mc E_{\mathrm{sparse}}^{(\gamma)}$ does not form a good relative-error design (see~\cref{lem:NotRelDesign}), the second-look principle is precisely what overcomes this lack of relative-error randomness and enables effective stabilizer-shadow estimation.
We now explain this mechanism. For a Pauli string $P$, let $\mathrm{wt}_X(P)$ denote the number of qubit positions on which $P$ acts as $X$ or $Y$. By choosing between $O$ and $H^{\otimes n}OH^{\otimes n}$, the protocol selects a representation whose stabilizer group exhibits a favorable $\mathrm{wt}_X(P)$ distribution; see the explicit distribution bound in~\cref{lem:PauliDist} of the
appendix. Consequently, the stabilizer components that are most sensitive to the imperfect  measurement channel of the sparse ensemble cannot appear in large numbers, and we carefully bound all their contributions by quantities that depend explicitly on $\mathrm{wt}_X(P)$.

We now outline the proof. The first step is to expand the second moment of the sparse Clifford-IQP ensemble in the Pauli basis. In this representation, the induced shadow measurement channel~\cref{Eq:idealchannel} becomes Pauli-diagonal, so that each Pauli input is rescaled by an explicit coefficient. The detailed decomposition is given below,
with the proof left to \cref{Ap:moment2Pauli}.

\begin{lemma}[Pauli decomposition of the second moment]\label{lem:PauliDiagonalChannel}
In the SPS protocol, the second moment of the sparse Clifford-IQP ensemble in~\cref{thm: sparse 2-moment}  admits a Pauli decomposition as
\begin{equation}\label{Eq:P2P}
{\mathbf{M}}_{\mc{E}_\text{sparse}^{(\gamma)}}^{(2)}=D^{-3}\sum_{P\in\mb{P}_n}\sigma_P\, P\otimes P .
\end{equation}
and the coefficient of the Pauli operator in the form $ {P} = \id_2^{\otimes n_1}\otimes  {Z}^{\otimes n_2}\otimes \{{X},{Y}\}^{\otimes n_3} $ reads
\begin{equation}\label{eq:sigmaPclosed}
    \sigma_P
    =
    \left[
    1+\left(1-\frac{2\gamma\ln n}{n}\right)^{n_3}
    \right]^{n_1}
    \left[
    1-\left(1-\frac{2\gamma\ln n}{n}\right)^{n_3}
    \right]^{n_2}.
\end{equation}
\end{lemma}

A useful feature of \cref{eq:sigmaPclosed} is that the coefficient $\sigma_P$ depends only on the Pauli weight data $(n_1,n_2,n_3)$, rather than on the locations of the corresponding Pauli operators, due to the symmetry of the ensemble $\mc{E}_\text{sparse}^{(\gamma)}$. In particular, $n_3$ is precisely the $X$-weight $\mathrm{wt}_X(P)$.

This Pauli decomposition implies that for a stabilizer-state observable $O$, only Pauli strings in $\mathrm{STAB}(O)$ contribute to the estimation bias, so the bias is controlled by the quantity of the form $D^{-1}\sum_{P\in\text{STAB}(O)}|\sigma_P-1|$. Thus, the proof reduces to understanding how the coefficient $\sigma_P$ behaves as a function of the Pauli weights, especially the $X$-weight $n_3$.

We then split the stabilizer Paulis $P$ of observable $O$ according to $n_3$.
For very small $n_3$, the deviation of $\sigma_P$ from its ideal value $1$ can be relatively large, but only a limited number of stabilizer Paulis can contribute, so the total contribution is controlled by some counting bound.
For large $n_3$, $\sigma_P$ becomes close to its ideal value $1$, so even if many stabilizer Paulis are present, each individual term is sufficiently suppressed.
For moderately small and intermediate $n_3$, the stabilizer-counting bound and the decay in the explicit coefficient $\sigma_P$ act together, making the total contribution negligible.
Summing these four regimes gives the claimed bound in \cref{Eq:Bias}.
The key reason this decomposition works is the second-look principle. By improving the $n_3$ distribution of the selected stabilizer representation, it prevents the two unfavorable effects, large deviations of $\sigma_P$ from its ideal value $1$ and many stabilizer Paulis, from occurring simultaneously.

We emphasize several additional features of \cref{th:PshadowMain} that are crucial for practice and sharpen the contrast to existing shallow-shadow constructions.
First, the bias bound in~\cref{Eq:Bias} is fully explicit. We give concrete coefficients so one can directly translate a target bias into an experimentally actionable circuit budget.
Second, the circuit cost is \emph{continuously tunable} through the sparsity parameter $\gamma$, which allows one to adjust the estimation bias smoothly as a function of the expected two-qubit gate count.
This stands in marked contrast to patch-based gluing schemes in Ref.~\cite{schuster2024random}, where the circuit resources in general could not be fine-grainedly calibrated for a fixed system size.
These advantages together establish SPS as a more direct and controllable route from theory-level guarantees to experiment-level implementation.

The variance bound is obtained from the third-moment formula of the sparse Clifford-IQP ensemble~\cref{Ap:moment3}.
At the end of \cref{Ap:MainVar}, we further discuss the total variance of the full estimator by taking into account the sample allocation $N_f$ and $N_d$ used in \cref{algo:PhaseShadow}, including the factor-of-two split required by the second-look principle.

We numerically validate the guarantees of \cref{th:PshadowMain} and assess their practical relevance. First, the simulations confirm that the analytical bias bound in \cref{Eq:Bias} is rigorous but conservative. While the theorem requires $\gamma>3$, the bias already decreases rapidly for significantly smaller values, around $\gamma\approx2$. As shown in \cref{fig:broader_use}(a), at $\gamma=2$ the estimation bias is reduced to about $0.05$ for most tested stabilizer states, far below the upper bound in \cref{th:PshadowMain}.

Second, although ~\cref{th:PshadowMain} only guarantees sub-exponential scaling of the estimation variance, the observed variance remains small and weakly dependent on system size in all tested regimes. As shown in \cref{fig:broader_use}(b), the statistical fluctuations are modest and do not pose a bottleneck, corroborating the empirical stability of the estimator.

\begin{figure}[htbp]
    \centering \includegraphics[width=\linewidth]{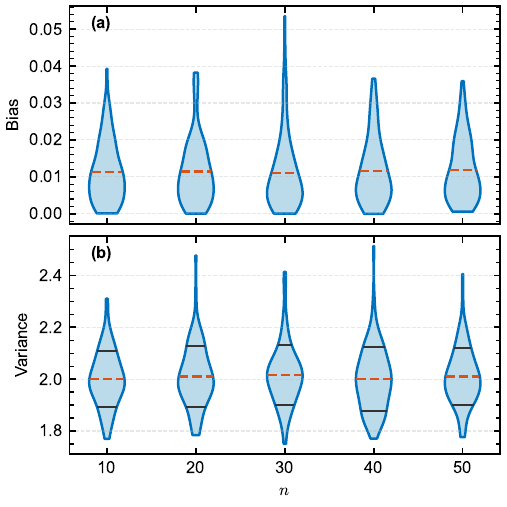} %
    \caption{Statistical distribution of the estimation bias and variance with respect to qubit number $n$. Data are collected from 100 independently sampled random stabilizer states with $\gamma = 2$ and $10^4$ shots for each state. The violin plots illustrate the variation in estimator performance among different quantum states. The dashed lines indicate the mean value of bias and variance, and the solid black lines in (b) represent the standard deviation range. }
    \label{fig:broader_use}
\end{figure}

Finally, we compare SPS with existing shallow-shadow schemes at fixed two-qubit gate cost. According to~\cref{app:gate_count_mapping}, $\gamma=2$ corresponds to about $N_g\approx182$ two-qubit gates when $n=48$; see \cref{tab:gate_count_comparison}. As demonstrated in \cref{fig:guarentee}, SPS reaches the statistical noise floor using several times fewer entangling gates across representative stabilizer states. Even in unfavorable cases where the second-look principle provides no additional benefit, as illustrated in \cref{fig:guarentee}(b), SPS remains competitive.

Taken together, these results show that the theoretical guarantees of SPS are not only valid but also practically meaningful: the bias bound is robust, the variance is well controlled, and the resource efficiency translates directly into improved experimental performance.

\begin{figure*}[htbp]
    \centering \includegraphics[width=\linewidth]{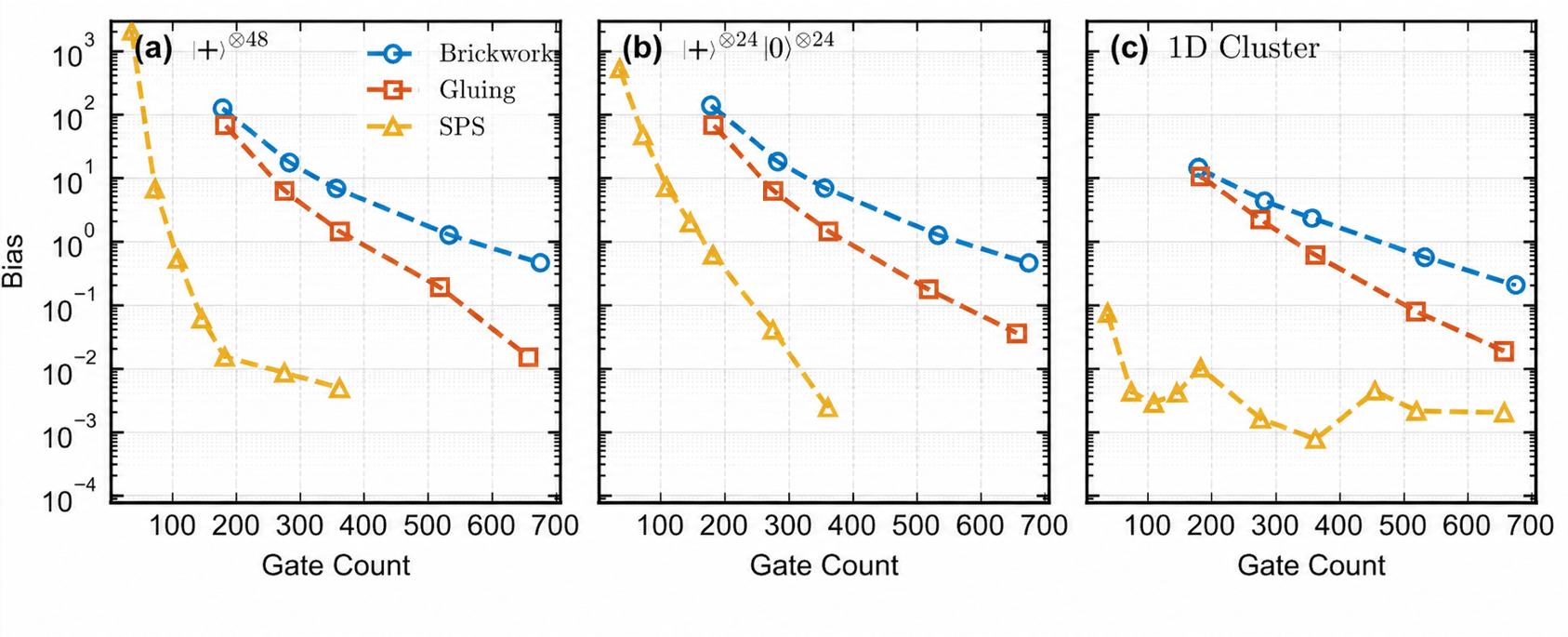} %

    \caption{Comparison of the estimation bias against two-qubit gate count (see \cref{app:gate_count_mapping}) for three shallow shadow schemes: 1D brickwork circuits (blue circles,~\cite{bertoni2024shallow}), Gluing circuits (orange squares,~\cite{schuster2024random}), and SPS (yellow triangles, our approach). Notice that the logarithmic-depth \emph{one-dimensional} (1D) brickwork baseline exhibits relative-error approximate state 2-design behavior established in Ref.~\cite{heinrich2025anti,dalzell2022random}. The results are shown for (a) the tensor product state vector $\ket{+}^{\otimes n}$, (b) $\ket{+}^{\otimes n/2}\ket{0}^{\otimes n/2}$, with $n=48$ and $2 \times 10^4$ shots and (c) the 1D cluster state. For SPS, the bias drops rapidly below the statistical noise floor ($\sim 10^{-2}$), making the measurement error the dominant source of uncertainty.  }
    \label{fig:guarentee}
\end{figure*}

\section{Exactly Unbiased Shallow Phase Shadow}\label{sec:unbiased}

The biased estimator introduced in the previous section achieves efficient reconstruction for stabilizer-state observables. In this section, we  construct an \emph{unbiased} estimator that exactly inverts the effective measurement channel induced by the sparse ensemble $\mc{E}_{\text{sparse}}^{(\gamma)}$.

\begin{prop}[Guarantees of exactly unbiased recovery of shallow phase shadow]\label{th:PshadowUnbiased}
Suppose that one conducts
shallow phase shadow estimation using the random unitary ensemble $\mc{E}_{{\text{sparse}}}^{(\gamma)}$, then the unbiased estimator of $\rho_f$ for a single-shot is given by
\begin{equation} \label{Eq:InvNoisychannelPauli}
\widehat{\rho_f}_{\text{sparse}} := \sum_{P \in \mb{P}_n/\mathcal{Z}_n} \sigma_{P}^{-1} \, \tr(\Phi_{U,\mathbf{b}} {P}) \, {P},
\end{equation}
where $\sigma_P$ is the Pauli
coefficient of the forward noisy channel from \cref{Eq:sigmaP}, ensuring $\mathbb{E}_{\{U,\mathbf{b}\}} \widehat{\rho_f}_{\text{sparse}} = \rho_f$.
\end{prop}

The proof of~\cref{th:PshadowUnbiased} is left to~\cref{Ap:ProofUnbiased}. Here, the unbiased estimator eliminates systematic bias entirely, at the expense of a moderate overhead in classical post-processing. The channel coefficient $\sigma_P$ can be efficiently computed, as it depends only on the Pauli weight data $\mathrm{wt}_X(P)$ and $\mathrm{wt}_Z(P)$ through~\cref{lem:PauliDiagonalChannel}.

Importantly, for the stabilizer-fidelity tasks considered here, this post-processing remains efficient. Concretely, given a stabilizer-state observable $O$, it suffices to evaluate
\begin{equation}
\tr(O\widehat{\rho_f}_{\mathrm{sparse}})= \sum_{P\in \mathbf P_n/\mathcal Z_n}\sigma_P^{-1}\tr(\Phi_{U,\mathbf b}P)\tr(OP).
\end{equation}
Thus, only Pauli strings that have nonzero overlap with both the target observable $O$ and the sampled snapshot $\Phi_{U,\mathbf b}=U^\dagger\ket{\mathbf b}\bra{\mathbf b}U$ contribute to the estimator. For stabilizer-state observables, these relevant Pauli supports can be generated and intersected efficiently using stabilizer tableaus.  The numerical results in~\cref{app:efficiency} show that the post-processing time remains modest in the tested regimes, with approximately polynomial growth in system size.

Moreover, the coefficients $\sigma_P^{-1}$ may admit a structured tensor-network representation, such as a \emph{matrix-product-state} (MPS) or \emph{matrix-product-operator} (MPO) description over
Pauli strings~\cite{hu2025demonstration,bertoni2024shallow,farias2024robust}, which would allow $\tr(O\widehat{\rho_f}_{\mathrm{sparse}})$ to be evaluated by tensor contraction rather than explicit summation over all $P$. We leave a systematic investigation of such MPS/MPO representations for the sparse phase ensemble to future work.

Beyond post-processing efficiency, the Pauli-diagonal form of the inverse map provides a direct route to error mitigation. If the physical noise is described by a Pauli channel, it simply modifies the Pauli rescaling coefficients of the shadow channel~\cite{qin2025classical,zhang2025robust}. The ideal sparsity-induced coefficient can therefore be replaced by an effective coefficient that incorporates both shallow-circuit sparsity and the noise response. In this way, SPS naturally integrates Pauli-channel error mitigation within the same classical post-processing framework.

\section{Discussion}

In this work, we have introduced \emph{shallow phase shadows} (SPS), a constant-depth protocol for global shadow estimation of stabilizer-state fidelities. At a conceptual level, our results have overcome a common prejudice in shadow estimation: that retaining the ``measure first, ask questions later'' flexibility of classical shadows necessarily requires generic random measurements, whereas tailoring a measurement scheme to the quantities of interest inevitably sacrifices this flexibility. SPS has shown that this dichotomy is not necessary. By adapting the randomization to a structured class of observables, while leaving the choice of the particular observable to post-processing, we retained the defining flexibility of shadow estimation with substantially less randomness. Technically, this has become possible through the combination of a remarkably simple sparse Clifford-IQP $S$-$CZ$-$H$ ensemble and the \emph{second-look} principle. Despite not forming a relative-error approximate state 2-design, this ensemble enabled efficient global estimation of stabilizer-state fidelities. At the same time, its commuting diagonal structure allowed the quantum circuit to be implemented in constant depth through a space-time trade-off assisted by mid-circuit measurements and classical feedback, with the same asymptotic auxiliary-system overhead as the leading gluing framework. The resulting circuits required only sparse $CZ$ entangling gates and were therefore particularly simple from an implementation perspective. Taken together, these results separated the requirements for successful shadow estimation of structured observables from the substantially stronger requirement of generating generic approximate-design randomness. They thereby suggested a broader principle: scalable quantum readout need not reproduce generic randomness when the randomization can instead be meaningfully adapted to the class of quantities one ultimately intends to estimate.

Our results point to several compelling directions. First, while we have focused on stabilizer fidelities, the framework should extend naturally to observables carrying modest non-stabilizerness (magic). A promising next step is to integrate SPS with robustness-of-magic techniques~\cite{howard2017application} to obtain bias and sample-complexity bounds that scale explicitly with the magic content of the observable. Second, and more broadly, this work raises the question of whether the full randomness needed for relative-error approximate designs is genuinely necessary for other structured shadow-estimation tasks. Characterizing the gap between task-specific shadow guarantees and conventional design conditions may lead to new classes of shallow, task-adapted circuit ensembles. From a hardware standpoint, SPS highlights a practical route for all-to-all connected atomic and ion platforms: by investing spatial resources in the form of auxiliary systems,  one can effectively trade space for time, mitigating the cost of slow entangling gates. This suggests that scalability on such platforms can be advanced not only by incremental improvements in gate speed and fidelity, but also by deliberate protocol design that shifts the dominant cost from circuit depth to qubit layout.

\section{Acknowledgments}

This work has been supported by the Quantum Science and Technology-National Science and Technology Major Project Grant Nos.~2024ZD0301900 and~2021ZD0302000, the National Natural Science Foundation of China (NSFC) Grant No.~12575012, the Shanghai QiYuan Innovation Foundation, the Shanghai Municipal Commission of Science and Technology with Grant No.~25511103200, the Shanghai Science and Technology Innovation Action Plan Grant No. 24LZ1400200, the Shanghai Pilot Program for Basic Research-Fudan University~21TQ1400100 (25TQ003, collaborated with X.L.), the CCF-Quantum CTek Superconducting Quantum Computing CCF-QC2025006.
J.E.\ acknowledges the support from the BMFTR (QSolid, PasQuops, Hybrid++, QuSol), the
Quantum Flagship (Millenion, PasQuans2), Berlin Quantum, the Munich Quantum Valley, the DFG (SPP 2514, CRC 183), the Clusters of Excellence (MATH+ and ML4Q),
and the European Research Council (DebuQC).
X.L.\ also acknowledges the support from National Key Research and Development Program of
China (2021YFA1400900), Innovation Program for Quantum Science and
Technology of China (2024ZD0300100), National Natural Science Foundation
of China (11934002), and Shanghai Municipal Science and Technology Major Project (Grant No. 2019SHZDZX01, 24DP2600100).

\bibliography{Shallow}

\clearpage

\appendix
\onecolumngrid
\startappendixtoc
\vspace{3em}
\begin{center}
\textbf{\large Supplementary material: \\ \smallskip Constant-depth global shadow estimation}
\end{center}
\appendixtableofcontents
\clearpage

\begin{appendix}

\section{Notations and preliminaries}\label{App:UaA}

In this appendix, we summarize the notation and mathematical conventions employed throughout the proofs.

We begin by defining set-theoretic operations—specifically \textit{union}, \textit{intersection}, and \textit{inclusion}—as applied to Boolean matrices. Let $A$ and $B$ be two Boolean matrices (or operators) of compatible dimensions. Their union ($A \cup B$) and intersection ($A \cap B$) are defined via element-wise logical OR ($\vee$) and AND ($\wedge$) operations, respectively. Formally, the elements $c_{i,j}$ of the resulting matrix $C$ are given by
\begin{eqnarray}
&\text{Union: } C = A \cup B \implies c_{i,j} = a_{i,j} \vee b_{i,j},\\
&\text{Intersection: } C = A \cap B \implies c_{i,j} = a_{i,j} \wedge b_{i,j},
\end{eqnarray}
where $a_{i,j}$ and $b_{i,j}$ denote the corresponding elements of $A$ and $B$.
We define element inclusion such that $\ket{i}\bra{j} \in A$ if and only if the corresponding matrix element satisfies $\bra{i}A\ket{j} = 1$ (i.e., $a_{i,j} = 1$). Furthermore, for matrices $A$ and $B$ of compatible dimensions, the “subset relationship” $A \subseteq B$ holds if, for all indices $(i,j)$, $a_{i,j} = 1$ implies $b_{i,j} = 1$.

For an $n$-qubit Pauli string $P$, we let $\mathrm{wt}_X(P)$ denote the number of Pauli $X$ and $Y$ operators in the string, and $\mathrm{wt}_Z(P)$ denote the number of Pauli $Z$ operators. For a Pauli operator of the form $P = \mathbb{I}_2^{\otimes n_1} \otimes Z^{\otimes n_2} \otimes \{X,Y\}^{\otimes  n_3}$, it follows that $\mathrm{wt}_X(P) = n_3$ and $\mathrm{wt}_Z(P) = n_2$.
In what follows,
we adopt the
notation $[n] := \{0, 1, \dots, n-1\}$. For any integer $k \in \{0, \dots, n-1\}$, $\mathcal{I}_k$ represents the collection of all $k$-element subsets of $[n]$. Given a subset $I \subseteq [n]$ and an operator $A$ acting on a single qubit, $A^I$ denotes the tensor product resulting from applying $A$ to every qubit indexed by $I$.
Throughout this work, $\mathrm{negl}(n)$ refers to any negligible function of the integer $n$; specifically, a function satisfying $\mathrm{negl}(n) = o(n^{-c})$ for every constant $c > 0$.

Finally, let $\pi$ be an element of the symmetric group $S_m$ of order $m$. Its unitary representation on $\mathcal{H}_D^{\otimes m}$ is denoted by $V_n(\pi)$, acting as
\begin{equation}
V_n(\pi) \bigotimes_{i=1}^{m} \ket{\psi_{i}} = \bigotimes_{i=1}^{m} \ket{\psi_{\pi^{-1}(i)}},
\end{equation}
where $\{\psi_1, \psi_2, \dots, \psi_m\}$ are $n$-qubit quantum states.

\section{Proof of Theorem~\ref{th:constdepth}}\label{Ap:proofmain}

In this section, we first state the formal version of \cref{th:constdepth}, including the explicit success probabilities inherited from
\cref{prop:resource_limits_formal-ap}
(the formal version of \cref{prop:resource_limits}).

\begin{theorem}[Shallow phase shadow in constant depth, formal]\label{th:constdepth_formal}
Run the shallow phase shadow protocol on the all-to-all architecture, guaranteeing that for every stabilizer-state observable $O$, the resulting estimator has bias at most a fixed value  $\epsilon\in(0,1)$. Then,  one has the implementation with
\begin{itemize}
\item  constant-depth, using
    $7.5n\ln(6n^2/\epsilon)$ auxiliary qubits,  with probability at least $1-\exp[-\frac{5}{12} (n-1)\ln \frac{6n^2}{\epsilon}]$,
\item $5\ln (6n^2/\epsilon)$-depth without using auxiliary systems,   with probability at least $1-\sqrt{\epsilon}n^{-0.25}$.
\end{itemize}
\end{theorem}
\begin{proof}
    The proof of \cref{th:constdepth_formal} combines the circuit-resource bounds in \cref{prop:resource_limits_formal-ap} with the performance guarantee of SPS in \cref{th:PshadowMain}: one first chooses the sparsity parameter $\gamma$ so that the SPS estimator has bias at most $\epsilon$, and then translates this choice of $\gamma$ into the stated depth and auxiliary system requirements.

By Theorem~\ref{th:PshadowMain}, for $\gamma>3$ the second-look estimator has
\begin{equation}
\mathrm{Bias}(\widehat{o_f})\le (4n^2+2n)\,n^{-0.4\gamma}+\mathrm{negl}(n).
\end{equation}
Choose $\gamma=\frac{\ln(6n^2/\epsilon)}{0.4\ln n}$, then $n^{-0.4\gamma}=\epsilon/(6n^2)$, hence
\begin{equation}
\mathrm{Bias}(\widehat{o_f})
\le (4n^2+2n)\frac{\epsilon}{6n^2}+\mathrm{negl}(n)
\le \epsilon.
\end{equation}
Moreover, this choice satisfies $\gamma>3$ automatically since $\epsilon\in(0,1)$.

It remains to bound implementation resources. These follow directly from
\cref{prop:resource_limits_formal-ap} applied to the above value of $\gamma$:
(i) in the constant-depth realization with auxiliary qubits, with probability at least
\begin{equation}
1-\exp[-\frac{\gamma (n-1)\ln n}{6}]=1-\exp[-\frac{1}{2.4} (n-1)\ln \frac{6n^2}{\epsilon}],
\end{equation}
the number of auxiliary qubits is at most
$3\gamma n\ln n\le 7.5n\ln(6n^2/\epsilon)=O(n\log\frac{n}{\epsilon})$;
(ii) in the auxiliary-system-free realization, with probability at least $1-n^{1-\gamma/4}$ the two-qubit depth is at most
$2\gamma\ln n=5\,\ln(6n^2/\epsilon)=O(\log\frac{n}{\epsilon})$. In both cases, the stated resource bounds hold with overwhelmingly high probability for the properly chosen value of $\gamma$.
Finally, rewriting the failure probability using the chosen $\gamma$ gives
\begin{equation}
n^{1-\gamma/4}
= n\exp\!\Big(-\frac{\ln(6n^2/\epsilon)}{1.6}\Big)
= n\Big(\frac{\epsilon}{6n^2}\Big)^{0.625}
\le \sqrt{\epsilon}\,n^{-0.25},
\end{equation}
which yields the stated success probability in the
auxiliary-system-free case.
\end{proof}

\section{Moment functions of sparse Clifford-IQP ensemble}

In this section, we derive the second- and third-order moment functions of the sparse Clifford-IQP ensemble, $\mb{M}_{\mc{E}_{\text{sparse}}^{(\gamma)}}^{(2)}$ and $\mb{M}_{\mc{E}_{\text{sparse}}^{(\gamma)}}^{(3)}$. In this work, the second moment is used to characterize the effective shadow channel and to bound the estimation bias of SPS, while the third moment is used for estimation variance.

The derivation is organized into three parts.
In~\cref{Ap:moment2}, we derive the tensor-product expression for the second moment and prove~\cref{thm: sparse 2-moment}.
In~\cref{Ap:moment2Pauli}, we rewrite this second moment in the Pauli basis, obtaining the channel coefficients used in the bias analysis.
In~\cref{Ap:moment3}, we derive the corresponding third-moment formula, which will later enter the variance analysis.
The use of these moment functions within the SPS framework is discussed in the subsequent sections.
Specifically, the results of~\cref{Ap:moment2,Ap:moment2Pauli} are used in the proof of~\cref{th:PshadowMain} in~\cref{Ap:MainBias}, while~\cref{Ap:moment3} is used in the variance analysis in~\cref{Ap:MainVar}.

To this end, we first analyze the structural form of the ensemble $\mc{E}_{\text{sparse}}^{(\gamma)}$. According to \cref{Eq:diagonalcircuits}, the unitary element $U\in \mc{E}_{\text{sparse}}^{(\gamma)}$ is taken in the form that $U={H}^{\otimes n}U_SU_C\in \mc{E}_{\text{sparse}}^{(\gamma)}$. In this setting, we represent the moment function of the sparse ensemble as
\begin{equation}\label{Eq:momentHSCZ}
\begin{split}
\mb{M}_{\mc{E}_{\text{sparse}}^{(\gamma)}}^{(m)} &= D^{-1}\mathbb{E}_{U\sim \mc{E}_{\text{sparse}}^{(\gamma)}}\sum_{\mb{b}} (U^{\dagger}\ket{\mb{b}}\bra{\mb{b}}U)^{\otimes m}\\
    &= \mathbb{E}_{U_C} U_C^{\dagger,\otimes m}[D^{-1}\mathbb{E}_{U_S}\sum_{\mb{b}} (U_S^{\dagger}H^{\otimes n}\ket{\mb{b}}\bra{\mb{b}}H^{\otimes n}U_S)^{\otimes m}]U_C^{\otimes m}.\\
\end{split}
\end{equation}
This decomposition conveniently separates the expectation over the single-qubit gates ($U_S$) from the two-qubit gates ($U_C$).

Averaging over the random single-qubit diagonal layer $U_S$ enforces a local phase-cancellation condition on each qubit across the $m$ copies. As a result, only those matrix elements satisfying the condition survive. Following Ref.~\cite{zhang2025robust}, one obtains
\begin{equation}\label{Eq:momentHS}
    D^{-1}\, \mathbb{E}_{U_S} \sum_{\mb{b}}
    \big(U_S^{\dagger} H^{\otimes n}\ket{\mb{b}}\bra{\mb{b}}H^{\otimes n} U_S\big)^{\otimes m}
    = D^{-m}\sum_{(\hat{\mb{x}},\hat{\mb{y}})\in C^{(m)}}
    \ket{\hat{\mb{x}}}\bra{\hat{\mb{y}}},
\end{equation}
where the index set $C^{(m)}$ specifies the subset of $m$-copy computational-basis patterns that survive after averaging over the single-qubit diagonal phases.
Explicitly (see Appendix~B of Ref.~\cite{zhang2025robust}), $C^{(m)}$ is defined as
\begin{equation}\label{Eq:CmDef}
    C^{(m)} :=
    \big\{ (\hat{\mb{x}}, \hat{\mb{y}})
    \in (\{0,1\}^{mn})^2
    \;\big|\;
    \forall\, q\in[n],\;
    \sum_{t=1}^m
    (x_{q,t} - y_{q,t}) \equiv 0 \pmod{4}
    \big\},
\end{equation}
that is, for each qubit index $q$, the total bit differences across all $m$ copies must vanish modulo 4.

\subsection{Proof of  \cref{thm: sparse 2-moment}}\label{Ap:moment2}
We now specialize the general framework to the $m=2$ case to prove \cref{thm: sparse 2-moment}. When $m=2$, we denote the two-copy operators as $ \ket{\mb{\hat{x}}} = \ket{\mb{x}}\otimes \ket{\mb{w}} = \ket{\mb{x,w}}$, and $ \bra{\mb{\hat{y}}} = \bra{\mb{y}}\otimes \bra{\mb{s}} = \bra{\mb{y,s}}$, where $(\mb{x,w,y,s})$ are all $n$-bit binary strings. Then  \cref{Eq:CmDef} becomes

\begin{equation}\label{Eq:CmDef2}
    \begin{split}
        C^{(2)} &= \{(\mb{x},\mb{w},\mb{y},\mb{s})|x_l+w_l\equiv y_l+s_l\pmod4, \forall l\in [n]\}.\\
    \end{split}
\end{equation}
Since $x_l+w_l$ and $y_l+s_l$ take values in $\{0,1,2\}$, the congruence modulo $4$ is equivalent to equality over the integers.
Thus, the surviving elements $\ket{\mb{x,w}}\bra{\mb{y,s}}$ are those with equal local Hamming sums.

Substituting \cref{Eq:CmDef2} into \cref{Eq:momentHS} and using Eq.~(C20) of Ref.~\cite{zhang2025robust},
the resulting one-qubit operators can be grouped into three mutually orthogonal components:
$\Delta_2 = \ket{0,0}\bra{0,0}+\ket{1,1}\bra{1,1}$,
$\mathbb{I}_4 - \Delta_2 = \ket{0,1}\bra{0,1}+\ket{1,0}\bra{1,0}$, and
$\mathbb{S}_2 - \Delta_2 = \ket{0,1}\bra{1,0}+\ket{1,0}\bra{0,1}$.
Hence, we obtain
\begin{equation}\label{Eq:moment2final}
\begin{split}
    D^{-1}\, \mathbb{E}_{U_S} \sum_{\mathbf{b}}
    \big(U_S^{\dagger} H^{\otimes n}\ket{\mathbf{b}}\bra{\mathbf{b}}H^{\otimes n} U_S\big)^{\otimes 2}
    &= D^{-2}
    \sum_{(\mathbf{x},\mathbf{w},\mathbf{y},\mathbf{s})\in C^{(2)}}
    \ket{\mathbf{x,w}}\bra{\mathbf{y,s}}\\
    &= D^{-2}\!\left[\Delta_2 + (\mathbb{I}_4-\Delta_2) + (\mathbb{S}_2-\Delta_2)\right]^{\otimes n}\\
    &= D^{-2}
    \sum_{\substack{I_i+I_j+I_k=[n]}}
    \Delta_2^{I_i}\!
    \otimes (\mathbb{I}_4-\Delta_2)^{I_j}
    \otimes (\mathbb{S}_2-\Delta_2)^{I_k},
\end{split}
\end{equation}
where $I_i\in\mc I_i, I_j\in\mc I_j, I_k\in\mc I_k $ are mutually disjoint index sets that partition the qubit register into regions associated with the three operator types above. Consequently, we write $i=|I_i|, j=|I_j|,k=|I_k|$.
Returning to \cref{Eq:momentHSCZ}, one represents the expectation over $U_C$ as

\begin{equation}\label{Eq:momentHSCZ2}
\begin{split}
\mb{M}_{\mc{E}_{\text{sparse}}^{(\gamma)}}^{(2)}
    &= D^{-2}\sum_{\substack{I_i+I_j+I_k=[n]}}\mathbb{E}_{U_C} U_C^{\dagger,\otimes 2}[
    \Delta_2^{I_i}\!
    \otimes (\mathbb{I}_4-\Delta_2)^{I_j}
    \otimes (\mathbb{S}_2-\Delta_2)^{I_k}]U_C^{\otimes 2}.\\
\end{split}
\end{equation}
We now analyze the effect of averaging over the random $CZ$ gates. It suffices to first consider a single gate $CZ_{i,j}$. The expectation of its application is

\begin{equation}
\begin{split}
 &\mathbb{E}_{CZ_{i,j}} CZ^{\dagger,\otimes 2}_{i,j}\ket{\mb{x,w}}\bra{\mb{y,s}}CZ_{i,j}^{\otimes 2}\\
    =&(1-p)\ket{\mb{x,w}}\bra{\mb{y,s}}+pCZ^{\dagger,\otimes 2}_{i,j}\ket{\mb{x,w}}\bra{\mb{y,s}}CZ_{i,j}^{\otimes 2}\\
    =&(1-2p T_{i,j}^{(2)})\ket{\mb{x,w}}\bra{\mb{y,s}}\\
    =&(1-\frac{2\gamma \ln n}{n} T_{i,j}^{(2)})\ket{\mb{x,w}}\bra{\mb{y,s}},
\end{split}
\end{equation}
where $T_{i,j}^{(2)}:=x_ix_j+w_iw_j-y_iy_j-s_is_j\pmod{2}$. A direct check of the three local sectors shows that the coefficient $T_{i,j}^{(2)} = 1$ if and only if $\ket{x_i,w_i}\bra{y_i,s_i}\in \id_4-\Delta_2$ and $\ket{x_j,w_j}\bra{y_j,s_j}\in \mbb{S}_2-\Delta_2$, or vice versa. In all other cases, the coefficient vanishes, i.e., $T_{i,j}^{(2)}=0$.

For an arbitrary operator $\ket{\mb{x,w}}\bra{\mb{y,s}}$ in the subspace $\Delta_2^{I_i}\otimes (\id_4-\Delta_2)^{I_j} \otimes (\mbb{S}_2-\Delta_2)^{I_k}$, the condition $T_{i,j}^{(2)} = 1$ (which picks up a non-trivial factor) occurs $jk$ times (once for each pair of indices $(i,j)$ where $i \in I_j$ and $j \in I_k$, or vice versa). Since all $CZ_{i,j}$ gates are independent, we multiply their contributions. Therefore, the coefficient $(1-\frac{2\gamma \ln n}{n} )$ is multiplied $jk$ times, yielding
\begin{equation}
    \begin{split}
\mb{M}_{\mc{E}_{\text{sparse}}^{(\gamma)}}^{(2)}
    &= D^{-2}\sum_{\substack{I_i,I_j,I_k}}\mathbb{E}_{U_C} U_C^{\dagger,\otimes 2}[
    \Delta_2^{I_i}\!
    \otimes (\mathbb{I}_4-\Delta_2)^{I_j}
    \otimes (\mathbb{S}_2-\Delta_2)^{I_k}]U_C^{\otimes 2}\\
    &=D^{-2} \sum_{I_i,I_j,I_k} [1-\frac{2\gamma \ln n}{ n}]^{jk}\Delta_2^{I_i}\otimes (\id_4-\Delta_2)^{I_j} \otimes (\mbb{S}_2-\Delta_2)^{I_k},\\
\end{split}
\label{Eq: channel_noise}
\end{equation}
where the operator $\ket{\mathbf{x,w}}\bra{\mathbf{y,s}}\in\Delta_2^{\otimes I_i}\otimes (\id_2^{\otimes 2}-\Delta_2)^{\otimes I_j} \otimes (\mbb{S}_2-\Delta_2)^{\otimes I_k}$. Thus, the proof of \cref{thm: sparse 2-moment} is completed.

\subsection{The Pauli form of the second moment function }\label{Ap:moment2Pauli}

In this subsection, we rewrite the second moment $\mathbf{M}_{\mc{E}_\text{sparse}^{(\gamma)}}^{(2)}$
in the two-copy Pauli basis. This Pauli-basis form makes the induced shadow channel diagonal in the Pauli basis,
which is the representation used later in the bias analysis and in the construction of the unbiased estimator.

More specifically, the coefficient $\sigma_P$ in the Pauli basis derived below will play two distinct roles.
First, in \cref{Ap:MainBias}, it quantifies how far the SPS channel deviates from the ideal one on each Pauli component.
Second, in \cref{th:PshadowUnbiased}, the same coefficient determines the exact channel inversion formula.
We therefore proceed in two steps:
we first expand the elementary two-copy blocks
$\Delta_2$, $\id_4-\Delta_2$, and $\mathbb{S}_2-\Delta_2$
in the Pauli basis,
and then reorganize the resulting sum into a closed expression for the Pauli coefficient $\sigma_P$.

 As shown in  \cref{thm: sparse 2-moment}, ${\mathbf{M}}_{\mc{E}_\text{sparse}^{(\gamma)}}^{(2)}$ can be expressed as a  summation of tensor products of $\Delta_2$, $\id_4 - \Delta_2$, and $\mbb{S}_2 - \Delta_2$. These components can be written as a linear combination of 2-copy Pauli operators as
\begin{equation}
    \Delta_2 = \frac{1}{2}(\id_2 \id_2 + {ZZ}), \quad \id_4 - \Delta_2 = \frac{1}{2}(\id_2 \id_2 - {ZZ}), \quad \mbb{S}_2 - \Delta_2 = \frac{1}{2}({XX} + {YY}).
\end{equation}
This implies that ${\mathbf{M}}_{\mc{E}_\text{sparse}^{(\gamma)}}^{(2)}$ is given by a linear combination of 2-copy Pauli operators
\begin{equation}\label{Eq: P2P}
{\mathbf{M}}_{\mc{E}_\text{sparse}^{(\gamma)}}^{(2)} =D^{-3} \sum_{P\in\mb{P}_n} \sigma_P {P} \otimes {P}.
\end{equation}
Here, we aim to compute the detailed value of $\sigma_P$. Given a Pauli operator in the form ${P} =  \id_2^{I_{n_1}} \otimes {Z}_2^{I_{n_2}}\otimes \{{X},{Y}\}^{ I_{n_3}}$, where $I_{n_1}\in \mc{I}_{n_1}, I_{n_2}\in \mc{I}_{n_2}, I_{n_3}\in \mc{I}_{n_3}$ and $I_{n_1}+I_{n_2}+I_{n_3}= [n]$, we can compute $\sigma_P$ by taking the trace against $P \otimes P$,
\begin{equation}\label{Eq:sigmaP1}
\begin{split}
         \sigma_{P} &=D\tr({\mathbf{M}}_{\mc{E}_\text{sparse}^{(\gamma)}}^{(2)} \text{ }{P}\otimes {P} )\\
        &= D^{-1}\sum_{I_{i},I_{j},I_{k}:I_{i}+I_{j}+I_{k}=[n]} [1-\frac{2\gamma\ln n}{n}]^{jk}\tr\left[\left(\Delta_2^{I_{i}}\otimes (\id_4-\Delta_2)^{I_{j}} \otimes (\mbb{S}_2-\Delta_2)^{I_{k}}\right)\left( {P}\otimes {P}\right)\right]\\
        &= D^{-1}\sum_{I_{k}=I_{n_3},I_{i}+I_{j}=I_{n_1}+I_{n_2}} [1-\frac{2\gamma\ln n}{n}]^{jk}\tr[\left(\Delta_2^{I_{i}}\otimes (\id_4-\Delta_2)^{I_{j}} \otimes (\mbb{S}_2-\Delta_2)^{I_{k}}\right)\left( {P}\otimes {P}\right)]\\
        &=\sum_{I_{k}=I_{n_3},I_{i}+I_{j}=I_{n_1}+I_{n_2}}[1-\frac{2\gamma\ln n}{n}]^{jn_3}(-1)^{|I_j\cap I_{n_2}|}.
\end{split}
\end{equation}
To reorganize the sum, let $s:=|I_j|$, namely the number of qubit positions in $I_{n_1}\cup I_{n_2}$ assigned to the factor $(\id_4-\Delta_2)$. Among these $s$ qubit positions, let $t:=|I_j\cap I_{n_2}|$ be the number of qubit positions coming from the $Z$-support of $P$. Then $s-t$ qubit positions are chosen from the identity support $I_{n_1}$,
\begin{equation}\label{Eq:sigmaP}
\begin{split}
         \sigma_{P}
        &=\sum_{I_{k}=I_{n_3},I_{i}+I_{j}=I_{n_1}+I_{n_2}}[1-\frac{2\gamma\ln n}{n}]^{jn_3}(-1)^{|I_j\cap I_{n_2}|}\\
        &=\sum_{s=0}^{n_1+n_2}\sum_{I_{n_1+n_2-s}, I_s: I_{n_1+n_2-s}+I_s = I_{n_1}+I_{n_2}}[1-\frac{2\gamma\ln n}{n}]^{sn_3}(-1)^{|I_{n_2}\cap I_s|}\\
        &=\sum_{s=0}^{n_1+n_2}\sum_{t:0\leq t\leq n_2, 0\leq s-t\leq n_1}(-1)^t C_{n_2}^t C_{n_1}^{s-t} [1-\frac{2\gamma\ln n}{n}]^{sn_3}\\
        &=[1+(1-\frac{2\gamma \ln n}{n})^{n_3}]^{n_1}[1-(1-\frac{2\gamma \ln n}{n})^{n_3}]^{n_2}.
\end{split}
\end{equation}
The last equality follows from the binomial theorem, equivalently from the generating-function identity of Krawtchouk polynomials.

\subsection{The third-order moment function }\label{Ap:moment3}

Having fully characterized the second moment function, we now turn to the third-order case ($m=3$), which is essential for bounding the estimation variance in \cref{Ap:MainVar}.

\begin{lemma}[Third moment functions of
the sparse ensemble]\label{thm: sparse 3-moment}
The third moment function of the sparse Clifford-IQP ensemble $\mc{E}_{\text{sparse}}^{(\gamma)}$  is
\begin{equation}
    \mb{M}_{\mc{E}_{\text{sparse}}^{(\gamma)}}^{(3)}=D^{-3}\sum_{\mb{x+w+z=y+s+t}}(1-\frac{2\gamma\ln n}{n})^{N_{\mb{x,w,z,y,s,t}}}\ket{\mb{x,w,z}}\bra{\mb{y,s,t}},
\label{Eq:moment3}
\end{equation}
where $N_{\mb{x,w,z,y,s,t}}\coloneq \sum_{i<j}\mb{1}\{x_ix_j+w_iw_j+z_iz_j\neq y_iy_j+s_is_j+t_it_j\pmod{2}\}$.
\end{lemma}
 \begin{proof}
The proof follows a similar structure to $m=2$. When $m=3$, we denote the three-copy operators as $\ket{\mb{\hat{x}}} = \ket{\mb{x}}\otimes\ket{\mb{w}}\otimes\ket{\mb{z}} = \ket{\mb{x,w,z}}$, and $\bra{\mb{\hat{y}}} = \bra{\mb{y}}\otimes\bra{\mb{s}}\otimes\bra{\mb{t}} = \bra{\mb{y,s,t}}$, where $(\mb{x,w,z,y,s,t})$ are all $n$-bit binary strings. Then \cref{Eq:CmDef} becomes
\begin{equation}\label{Eq:CmDef3}
\begin{split}
C^{(3)} &= \{(\mb{x,w,z,y,s,t})\,|\, x_l+w_l+z_l \equiv y_l+s_l+t_l \pmod{4}, \, \forall l\in [n]\}.
\end{split}
\end{equation}
Since both sides of the congruence take values in $\{0,1,2,3\}$, the congruence modulo $4$ is equivalent to equality over the integers at each qubit position. Hence, the inner expectation over $U_S$ (from \cref{Eq:momentHS}) yields
\begin{equation}\label{Eq:moment3final}
\begin{split}
D^{-1}\,\mathbb{E}_{U_S}\sum_{\mathbf{b}}
\big(U_S^{\dagger}H^{\otimes n}\ket{\mathbf{b}}\bra{\mathbf{b}}H^{\otimes n}U_S\big)^{\otimes 3}
&=D^{-3}\sum_{(\mathbf{x,w,z,y,s,t})\in C^{(3)}}\ket{\mathbf{x,w,z}}\bra{\mathbf{y,s,t}}\\
&=D^{-3}\sum_{\mb{x+w+z=y+s+t}}\ket{\mathbf{x,w,z}}\bra{\mathbf{y,s,t}}\\
\end{split}
\end{equation}
Returning to \cref{Eq:momentHSCZ}, one calculate the expectation over $U_C$ to arrive at
\begin{equation}\label{Eq:momentHSCZ3}
\begin{split}
\mb{M}_{\mc{E}_{\text{sparse}}^{(\gamma)}}^{(3)}
&=D^{-3}\sum_{\mb{x+w+z=y+s+t}}\mathbb{E}_{U_C}
U_C^{\dagger,\otimes 3}\ket{\mathbf{x,w,z}}\bra{\mathbf{y,s,t}}U_C^{\otimes 3}.
\end{split}
\end{equation}
Next, we analyze the effect of random $CZ$ gates. Consider a single random $CZ$ gate first.
\begin{equation}
\begin{split}
\mathbb{E}_{CZ_{i,j}}CZ_{i,j}^{\dagger,\otimes 3}\ket{\mb{x,w,z}}\bra{\mb{y,s,t}}CZ_{i,j}^{\otimes 3}
&=(1-p)\ket{\mb{x,w,z}}\bra{\mb{y,s,t}}+p\,CZ_{i,j}^{\dagger,\otimes 3}\ket{\mb{x,w,z}}\bra{\mb{y,s,t}}CZ_{i,j}^{\otimes 3}\\
&=(1-2p\,T_{i,j}^{(3)})\ket{\mb{x,w,z}}\bra{\mb{y,s,t}}\\
&=(1-\frac{2\gamma\ln n}{n}\,T_{i,j}^{(3)})\ket{\mb{x,w,z}}\bra{\mb{y,s,t}},
\end{split}
\end{equation}
where $T_{i,j}^{(3)}\in\{0,1\}$ is defined by $T_{i,j}^{(3)}:=x_ix_j+w_iw_j+z_iz_j-y_iy_j-s_is_j-t_it_j\pmod{2}$. For a general operator $\ket{\mb{x,w,z}}\bra{\mb{y,s,t}}$, the factor $(1-\frac{2\gamma\ln n}{n})$ is multiplied exactly $N_{\mb{x,w,z,y,s,t}}$ times (as defined in the lemma statement) over the expectation of all $C_n^2$ possible $CZ$ gates. This yields
\begin{equation}\label{Eq:channel_noise3}
\begin{split}
D^{-3}\sum_{\mb{x+w+z=y+s+t}}\mathbb{E}_{U_C}
U_C^{\dagger,\otimes 3}\ket{\mathbf{x,w,z}}\bra{\mathbf{y,s,t}}U_C^{\otimes 3}=D^{-3}\sum_{\mb{x+w+z=y+s+t}}(1-\frac{2\gamma\ln n}{n})^{N_{\mb{x,w,z,y,s,t}}}\ket{\mb{x,w,z}}\bra{\mb{y,s,t}}.
\end{split}
\end{equation}
Thus, the proof of \cref{thm: sparse 3-moment} is completed.
\end{proof}

\section{Proof of Proposition~\ref{prop:resource_limits}}\label{Ap:ProofProp1}

For completeness, we restate the resource estimate with its explicit success probabilities before giving the proof.

\begin{prop}[Resource estimation of quantum circuit realization]\label{prop:resource_limits_formal-ap}
For the sparse unitary ensemble $\mc{E}_{\text{sparse}}^{(\gamma)}$ defined around \cref{Eq:diagonalcircuits}, the quantum circuit realization satisfies the following resource bounds.
\begin{itemize}
    \item[(i)] With auxiliary systems: The protocol runs in constant depth using the number of auxiliary qubits bounded by $ 3\gamma n \ln n$, with probability at least $1 - \exp\left[-\frac{\gamma (n-1)\ln n}{6}\right]$;
    \item[(ii)] Without auxiliary systems: The protocol runs in logarithmic depth, bounded  by $2\gamma \ln n$, with probability at least $1 - n^{1-\gamma/4}$.
\end{itemize}
\end{prop}

\begin{proof}
\textit{Case (i):} For $c=2$, \cref{eq:ancilla} gives $n_{\mathrm{tot}}=3|E|$,
where $|E| \sim \mathrm{Bin}\!\left(C_n^2,p\right)$ with
$p=\gamma \ln n/n$. Let
\begin{equation}
\mu_E:=\mathbb{E}[|E|]=pC_n^2
= \frac{\gamma (n-1)\ln n}{2}
\le \frac{1}{2}\gamma n\ln n.
\end{equation}
Applying the multiplicative Chernoff bound to $|E|$, we obtain
\begin{equation}
\Pr\!\left(|E| \ge 2\mu_E\right)\le \exp\!\left(-\frac{\mu_E}{3}\right).
\end{equation}
Therefore,
\begin{equation}
\Pr\!\left(n_{\mathrm{tot}} \ge 6\mu_E\right)
=
\Pr\!\left(3|E| \ge 6\mu_E\right)
=
\Pr\!\left(|E| \ge 2\mu_E\right)
\le \exp\!\left(-\frac{\mu_E}{3}\right).
\end{equation}
Using $6\mu_E = 3\gamma (n-1)\ln n \le 3\gamma n\ln n$, the event $\{n_{\mathrm{tot}}\ge 3\gamma n\ln n\}$ is contained in $\{n_{\mathrm{tot}}\ge 6\mu_E\}$. Hence, we conclude that
\begin{equation}
\Pr\!\left(n_{\mathrm{tot}} \ge 3\gamma n\ln n\right)
\le
\exp\!\left(-\frac{\gamma (n-1)\ln n}{6}\right).
\end{equation}
Since $n_{\mathrm{anc}}<n_{\mathrm{tot}}$, the same bound also applies to the
auxiliary system count up to the same threshold, which proves the claim.

\textit{Case (ii):} Without
auxiliary qubits, the circuit depth is bounded by $\Delta(G) + 1$, where $\Delta(G)$ is the maximum degree of a graph. Standard results for sparse random graphs (e.g., Lemma 9 in Ref.~\cite{bremner2017achieving}) establish that for $p = \gamma \ln n/n$,
we have that
\begin{equation}
    \Pr[\Delta(G) \ge 2\gamma \ln n] \le n^{1-\gamma/4}.
\end{equation}
Thus, the depth is bounded by $2\gamma \ln n$ with high probability.
\end{proof}

\section{Details on the second-look principle}\label{Ap:DetailHadamard}

This section provides the technical details behind the second-look principle introduced in the main text.
The section is organized as follows.
In \cref{Ap:ProofOfHt1}, we prove the lower bound in \cref{Eq:exHadatrick_1} for the original task.
In \cref{Ap:NotRelDesign}, we use the same example to show that the sparse Clifford-IQP ensemble is not a relative-error approximate state $2$-design with arbitrarily small error.
Finally, in \cref{Ap:weight}, we use the stabilizer tableau to control how Pauli weights are distributed inside the stabilizer group, up to a possible global Hadamard conjugation. This is the precise point where the second-look principle enters the proof.
Taken together, these three steps identify the obstruction faced by the default SPS track and motivate the second-look principle used in the SPS protocol.

\subsection{Proof of \cref{Eq:exHadatrick_1}}
\label{Ap:ProofOfHt1}

In order to address the proof, we start by considering the example
\begin{equation}\label{eq:HadamardPsi}
\ket{\psi}=\ket{+}^{\otimes 2}\otimes \ket{0}^{\otimes (n-2)},\qquad
\rho = O = \ket{\psi}\!\bra{\psi}.
\end{equation}
One derives from~\cref{Eq:BiasDef} that the estimation bias of shadow estimation with $U\sim\mc E_{\text{sparse}}^{(\gamma)}$ is given by
\begin{equation}\label{Eq:Ht1BiasDecompRecall}
\mathrm{Bias}(\widehat{o_f})
=\sum_{I_i+I_j+I_k=[n],jk>0}\Bigl(1-\frac{2\gamma\ln n}{n}\Bigr)^{jk}
\tr\!\Bigl[
\rho\otimes O\,
\Delta_2^{I_i}\otimes(\id_4-\Delta_2)^{I_j}\otimes(\mathbb{S}_2-\Delta_2)^{I_k}
\Bigr],
\end{equation}
where the sum is over mutually disjoint sets $I_i\in \mc I_i,I_j\in \mc I_j,I_k\in \mc I_k$.

For our choice of $\rho=O=\ket{\psi}\!\bra{\psi}$ in~\cref{eq:HadamardPsi}, every summand in \eqref{Eq:Ht1BiasDecompRecall} is non-negative. Hence, we can lower-bound the bias by retaining any single nonzero term.
Here, we choose a concrete partition
\begin{equation}
I_j=\{0\},\qquad I_k=\{1\},\qquad I_i=[n]\setminus\{0,1\},
\end{equation}
or we swap $I_j$ and $I_k$. In both cases,  $j=k=1$ and thus $(1-\frac{2\gamma\ln n}{n})^{jk}=(1-\frac{2\gamma\ln n}{n})$.
This yields
\begin{equation}\label{Eq:Ht1PickTerm}
\begin{split}
    \mathrm{Bias}(\widehat{o_f})&\ge
\Bigl(1-\frac{2\gamma\ln n}{n}\Bigr)
\tr\!\Bigl[
(\ket{\psi}\!\bra{\psi})^{\otimes 2}\,
[(\id_4-\Delta_2)_{1}\otimes(\mathbb{S}_2-\Delta_2)_{2}\otimes\Delta_2^{\otimes(n-2)}+(\mbb{S}_2-\Delta_2)_{1}\otimes(\id_4-\Delta_2)_{2}\otimes\Delta_2^{\otimes(n-2)}]
\Bigr]\\
&=\Bigl(1-\frac{2\gamma\ln n}{n}\Bigr)\times 2\bra{+}^{\otimes 2}\mathbb{S}_2-\Delta_2\ket{+}^{\otimes 2}\bra{+}^{\otimes 2}\mathbb{\id}_4-\Delta_2\ket{+}^{\otimes 2}\bra{0}\Delta_2\ket{0}^{\otimes(n-2)}=\frac{1}{2}\Bigl(1-\frac{2\gamma\ln n}{n}\Bigr),
\end{split}
\end{equation}
which is exactly \cref{Eq:exHadatrick_1}.

\subsection{Proof of~\cref{lem:NotRelDesign}}
\label{Ap:NotRelDesign}

In this subsection, we show that the sparse Clifford-IQP ensemble $\mc E_{\text{sparse}}^{(\gamma)}$ is not a relative-error approximate state $2$-design with an arbitrarily small error.
To do so, we first show that the  definition of approximate state 2-design in~\cref{eq:RelativeDef} is equivalent to that for every positive semidefinite operator $A\ge 0$,
\begin{equation}\label{eq:RelativePSDTest}
\bigl|\tr[(\mathbf M_{\mc E_{\mathrm{sparse}}^{(\gamma)}}^{(2)}
-\mathbf M_{\mathrm{Haar}}^{(2)})A]\bigr|
\le
\epsilon\,\tr(\mathbf M_{\mathrm{Haar}}^{(2)}A).
\end{equation}
It therefore suffices to find one positive operator $A$ for which the relative deviation in \cref{eq:RelativePSDTest} is bounded below by a constant.
We choose the same example as in \cref{eq:HadamardPsi}, and set $A:=\rho\otimes O=(\ket{\psi}\!\bra{\psi})^{\otimes 2}\ge 0$. From \cref{thm: sparse 2-moment}, we obtain that for this choice of $\rho$ and $O$,
\begin{equation}
\begin{split}
    \tr\!\bigl(\mathbf M_{\mc E_{\mathrm{sparse}}^{(\gamma)}}^{(2)}A\bigr)& = D^{-2} \sum_{I_i+I_j+I_k=[n]}\Bigl(1-\frac{2\gamma\ln n}{n}\Bigr)^{jk}
\tr\!\Bigl[
\rho\otimes O\,
\Delta_2^{I_i}\otimes(\id_4-\Delta_2)^{I_j}\otimes(\mathbb{S}_2-\Delta_2)^{I_k}
\Bigr]\\
&= D^{-2}\bigl[\mathrm{Bias}(\widehat{o_f})+\tr\bigr(
A \text{ }(\id_4^{\otimes n}+\mbb{S}_2^{\otimes n}-\Delta_2^{\otimes n})\bigr)\bigl]\\
&\geq D^{-2}\bigl[\frac12\Bigl(1-\frac{2\gamma\ln n}{n}\Bigr) +1+1-\frac14\bigl] = D^{-2}(\frac94-\frac{\gamma\ln n}{n}),
\end{split}
\end{equation}
where the inequality is obtained using the result of \cref{Eq:Ht1PickTerm}. On the other hand, for the Haar ensemble, one has
\begin{equation}
\tr\!\bigl(\mathbf M_{\mathrm{Haar}}^{(2)}A\bigr)
=
\frac{\tr(A)+\tr(\mathbb S_2^{\otimes n}A)}{D(D+1)}
=
\frac{1+\tr(\rho O)}{D(D+1)}
=
\frac{2}{D(D+1)}.
\end{equation}
Therefore, one arrives at
\begin{equation}\label{eq:RelErrLowerBound}
\epsilon
\ge
\frac{
\tr\!\bigl[(\mathbf M_{\mc E_{\mathrm{sparse}}^{(\gamma)}}^{(2)}
-\mathbf M_{\mathrm{Haar}}^{(2)})A\bigr]
}{
\tr\!\bigl(\mathbf M_{\mathrm{Haar}}^{(2)}A\bigr)
}
\ge
\frac{\frac94-\frac{\gamma\ln n}{n}}{\frac{2D}{D+1}}-1.
\end{equation}
As $n\to\infty$, this lower bound tends to
\begin{equation}
\epsilon \ge \frac18-\frac{\gamma\ln n}{n}=\frac 18-o(1),
\end{equation}
for the sparse regime $p=\gamma\ln n/n$ with $\gamma$ being a constant. It means that the relative error is bounded away from zero by a constant and cannot be made
arbitrarily small within this sparse Clifford-IQP family. We conclude that
$\mc E_{\mathrm{sparse}}^{(\gamma)}$ is not an $\epsilon$-approximate state $2$-design in the
relative-error sense.

\subsection{Stabilizer Pauli-weight counting under the second-look principle}\label{Ap:weight}

In this subsection, we provide the following lemma, which is a key bound on the distribution of Pauli weights $(\mathrm{wt}_X(P),\mathrm{wt}_Z(P))$ within this stabilizer group.

\begin{lemma}[Pauli-weight bound within a stabilizer group]\label{lem:PauliDist}
Let $O$ be a stabilizer-state observable and
$O' = H^{\otimes n} O H^{\otimes n}$ its Hadamard-conjugated counterpart.
At least one of these two observables admits a parameter $\zeta\ge n/2$
such that the number of stabilizers $P$ in the corresponding stabilizer group satisfies the following bounds.
\begin{itemize}
    \item The number of stabilizers $P$ with $\mathrm{wt}_X(P) = k>0$ is upper bounded by $[F(\zeta,k)-1]\,2^{\,n-\zeta}$.
    \item The number of stabilizers with $\mathrm{wt}_X(P) = k>0$ and $\mathrm{wt}_Z(P) = k'$ is bounded by $2^k[F(\zeta,k)-1]\,F(n-\zeta,k')$.
\end{itemize}
\end{lemma}
Here, $F(a,b) := \sum_{i=0}^{b} C_a^i$ (with $C_a^b = 0$ for $b>a$).

\begin{proof}

Let  $O = V\ket{\mathbf{0}}\bra{\mathbf{0}} V^{\dagger}$ be a stabilizer state with $V$ being a Clifford unitary.  The tableau representation~\cite{aaronson2004improved} of $V$ is denoted as $ T(V)=
\begin{bmatrix}
A(V) & B(V)\\
C(V) & D(V)
\end{bmatrix}.$ For convenience, we set $\zeta\coloneq \text{rank}(C)$. In the beginning, we justify that $\zeta \ge n/2$ for either $O$ or $H^{\otimes n}OH^{\otimes n}$.  To this end, we consider the stabilizer-state observable $H^{\otimes n} O H^{\otimes n}$.  The corresponding tableau representation satisfies $T(H^{\otimes n}V) =
\begin{bmatrix}
B(V) & A(V) \\
D(V) & C(V)
\end{bmatrix}$. Applying the same reasoning as above gives the same bound with $\zeta = \operatorname{rank}[D(V)]$.
Since $\operatorname{rank}[C(V)] + \operatorname{rank}[D(V)] \ge \operatorname{rank}[C(V),\; D(V)] = n$, at least one of these two ranks must be no smaller than $n/2$.  Thus, without loss of generality, we henceforth take the selected representation $T(V)$ to satisfy
$\zeta\ge n/2$.

In the tableau representation $T(V)$, $n$ rows of the binary matrix $[C(V), D(V)]$ generate the stabilizer group of $O$. This particular generating set is not unique: applying any invertible row operation to $[C(V), D(V)]$ simply gives a different set of generators for the same stabilizer group, and hence leaves $O$ unchanged. We therefore reduce the stabilizer generator matrix $[C(V), D(V)]$ to \emph{reduced row-echelon form} (RREF) over $\mbb{F}_2$ as
\begin{equation}
    [C(V),\; D(V)] =
\begin{bmatrix}
C_1(V)_{\zeta\times n} & D_1(V)_{\zeta\times n} \\
\varnothing & D_2(V)_{(n-\zeta)\times n}
\end{bmatrix},
\end{equation}
where $\rank C_1(V)=\rank C(V)=\zeta$. Moreover, $C_1(V)$ and $D_2(V)$ are also in the RREF. Here, a pivot refers to the leading $1$ of a nonzero row in the RREF over $\mathbb F_2$. Thus $C_1(V)$ has $\zeta$ pivots, while $D_2(V)$ has $n-\zeta$ pivots. We also obtain that every pivot column of $D_2(V)$ has zero entries in the upper block $D_1(V)$.

We now use this RREF to bound the Pauli-weight distribution inside the stabilizer group of $O$.

First, consider the stabilizers $P_1$ generated by choosing $\zeta_1$ of the first $\zeta$ rows, together with an \emph{arbitrary} subset of the remaining $n-\zeta$ rows in $D_2(V)$.
There are $C_{\zeta}^{\zeta_1}\,2^{\,n-\zeta}$ such Paulis. Since these generators include $\zeta_1$ pivots from $C_1(V)$, we have
\begin{equation}
\mathrm{wt}_X(P_1)\ge \zeta_1.
\end{equation}
The stabilizers $P$ with $\mathrm{wt}_X(P)=k$ can arise only from choosing at most $k$ of the first $\zeta$ rows. Summing over all possible $\zeta_1\le k$ gives
\begin{equation}\label{eq:Nx}
N_X(k)\coloneq\#\{P:\mathrm{wt}_X(P)=k\}
\le \sum_{\zeta_1=1}^{k} 2^{\,n-\zeta}C_{\zeta}^{\zeta_1}
=[F(\zeta,k)-1]\,2^{\,n-\zeta},
\end{equation}
where $\#\{\cdot\}$ denotes the cardinality of the set.

\begin{figure}[t]
\begin{equation}
[C,\;D]=
\begin{array}{c|cccc|cccc}
 & 0&1&2&3 & 0&1&2&3\\
\hline
g_0 & 1&0&1&0 & 0&0&0&0\\
g_1 & 0&1&0&0 & 0&0&0&0\\
\hline
h_0 & 0&0&0&0 & 1&0&1&0\\
h_1 & 0&0&0&0 & 0&0&0&1
\end{array}
\end{equation}
\caption{
A four-qubit example illustrating the notation used in the stabilizer-counting argument.
Here $\zeta=\rank(C)=2$. The upper block $C_1$ has pivots in columns $0$ and $1$, while the lower block $D_2$ has pivots in columns $0$ and $3$ of the $D$ part.
Consider the choice of one row from the first $\zeta$ rows, namely $\zeta_1=1$ and $\widehat{\zeta}_1=(1,0)$, corresponding to selecting $g_0$.
The resulting $X$-support is determined by $(1,0,1,0)$, so
$t(\widehat{\zeta}_1)=2$.
Among the pivots of $D_2$, the pivot of $h_0$ lies on qubit $0$, which is already occupied by the $X$-support. It is therefore covered, contributing to $t'(\widehat{\zeta}_1)=1$.
After excluding this covered lower-row pivot, only the lower row $h_1$ remains available as an uncovered $Z$-pivot row. Thus, if one further chooses $\zeta_2=1$ uncovered lower row, the only possible choice is $h_1$, while the covered row $h_0$ may still be chosen or not chosen without increasing the guaranteed lower bound on the $Z$-weight.
This construction ensures that, in the subcase indexed by $(\widehat{\zeta}_1,\zeta_2)$, the selected Pauli strings satisfy $\mathrm{wt}_X(P)\ge \zeta_1$ and $\mathrm{wt}_Z(P)\ge \zeta_2$.
}
\label{fig:tableau_counting_example}
\end{figure}

Second, we bound the number of stabilizers with prescribed $X$- and $Z$-weights, namely those satisfying $\mathrm{wt}_X(P)=k$ and $\mathrm{wt}_Z(P)=k'$. We also illustrate the notation used in this paragraph with a simple example in~\cref{fig:tableau_counting_example}. Fix a choice of $\zeta_1$ rows from the first $\zeta$
rows, and denote this choice by a binary string
$\widehat{\zeta}_1\in\{0,1\}^{\zeta}$ with
$\mathrm{wt}(\widehat{\zeta}_1)=\zeta_1$.  Once this choice is fixed, the
$X$-support of the partially generated stabilizer is fixed. We denote its
X-weight by $t(\widehat{\zeta}_1)$.  Since the selected $\zeta_1$ rows contain
independent $X$-pivots, we have
\begin{equation}
t(\widehat{\zeta}_1)\ge \zeta_1 .
\end{equation}
We now choose generators from the remaining $n-\zeta$ rows.  Some of the
$Z$-pivots in the lower block $D_2(V)$ may lie on the qubit positions that are already
occupied by the $X$-support generated from the first $\zeta$ rows.  In such
situations, multiplying by a lower-row generator changes a $Z$ contribution into a
$Y$ contribution, and hence these $Z$-pivots are not counted in the final
$Z$-weight.  Let $t'(\widehat{\zeta}_1)\le t(\widehat{\zeta}_1)$ be the number of such covered $Z$-pivots among the lower rows.  After removing
these covered pivot positions, we may choose $\zeta_2$ rows from the remaining $n-\zeta-t'(\widehat{\zeta}_1)$ lower rows.  In this subcase
$(\widehat{\zeta}_1,\zeta_2)$, every generated stabilizer satisfies
\begin{equation}
    \mathrm{wt}_X(P)=t(\widehat{\zeta}_1)\ge \zeta_1,\,\,
\mathrm{wt}_Z(P)\ge \zeta_2 .
\end{equation}
Moreover, this subcase contains $2^{t'(\widehat{\zeta}_1)}
C_{n-\zeta-t'(\widehat{\zeta}_1)}^{\zeta_2}$
stabilizers, where the factor $2^{t'(\widehat{\zeta}_1)}$ accounts for the
arbitrary choices on the covered lower-row pivots.
The above subcases exhaust all stabilizers of $O$.  Therefore, the number of
stabilizers with $\mathrm{wt}_X(P)=k$ and $\mathrm{wt}_Z(P)=k'$ is bounded by
the total number of subcases $(t(\widehat{\zeta}_1),\zeta_2)$ with
$\zeta_1\le k$, $\zeta_2\le k'$, and $t(\widehat{\zeta}_1)=k$.
\begin{equation}\label{eq:Nxz}
\begin{aligned}
&N_{X,Z}(k,k')\coloneq\#\{P:\mathrm{wt}_X(P)=k,\ \mathrm{wt}_Z(P)=k'\}  \\
&\le
\sum_{\zeta_1=1}^{k}
\sum_{\substack{\widehat{\zeta}_1:\,
\mathrm{wt}(\widehat{\zeta}_1)=\zeta_1,t(\widehat{\zeta}_1)=k}}
\sum_{\zeta_2=0}^{\min\{k',\,n-\zeta-t'(\widehat{\zeta}_1)\}}
2^{t'(\widehat{\zeta}_1)}
C_{n-\zeta-t'(\widehat{\zeta}_1)}^{\zeta_2} \\
&\le
\sum_{\zeta_1=1}^{k}
\sum_{\substack{\widehat{\zeta}_1:\,
\mathrm{wt}(\widehat{\zeta}_1)=\zeta_1,t(\widehat{\zeta}_1)=k}}
\sum_{\zeta_2=0}^{k'}
2^{t'(\widehat{\zeta}_1)}
C_{n-\zeta}^{\zeta_2} \\
&\le
\sum_{\zeta_1=1}^{k}
\sum_{\substack{\widehat{\zeta}_1:\,
\mathrm{wt}(\widehat{\zeta}_1)=\zeta_1,t(\widehat{\zeta}_1)=k}}
\sum_{\zeta_2=0}^{k'}
2^{t(\widehat{\zeta}_1)}
C_{n-\zeta}^{\zeta_2} =
\sum_{\zeta_1=1}^{k} 2^k
\sum_{\substack{\widehat{\zeta}_1:\,
\mathrm{wt}(\widehat{\zeta}_1)=\zeta_1,t(\widehat{\zeta}_1)=k}}
\sum_{\zeta_2=0}^{k'}
C_{n-\zeta}^{\zeta_2} \\
&\leq
\sum_{\zeta_1=1}^{k}
2^k \sum_{\substack{\widehat{\zeta}_1:\,
\mathrm{wt}(\widehat{\zeta}_1)=\zeta_1}} F(n-\zeta,k') \\
&=
2^k\,[F(\zeta,k)-1]\,F(n-\zeta,k') .
\end{aligned}
\end{equation}
This completes the proof.
\end{proof}

During the proof of \cref{lem:PauliDist}, we observe that it is efficient to check which of $O$ or $H^{\otimes n} O H^{\otimes n}$ satisfies the desired Pauli-distribution property (simply by comparing the two ranks).  In this way, we can measure both $\rho$ and $H^{\otimes n}\rho H^{\otimes n}$ first, then ask questions about the observable $O$, which does not violate the "measure first, ask questions later" philosophy of shadow estimation.

\section{Proof of \cref{th:PshadowMain}}\label{Ap:ProofOfMainThm}

This section aims at proving \cref{th:PshadowMain}, which provides the central bounds on the performance of the SPS protocol.
The proof is divided into two parts. First, we bound the estimation bias in \cref{Ap:MainBias}. Second, we bound the estimation variance in \cref{Ap:MainVar}.

\subsection{Estimation bias}\label{Ap:MainBias}

In this subsection, we prove the bias bound in \cref{th:PshadowMain}. The proof has two conceptual steps.
First, in \cref{lem:BiasPauli}, we express the estimation bias through the Pauli coefficients $\sigma_P$ defined in~\cref{Eq:sigmaP}. These coefficients characterize the induced measurement channel in the Pauli basis, so the channel error is completely encoded in the deviations of $\sigma_P$ from its ideal value.
Second, we combine the channel coefficients and the stabilizer weight bound in~\cref{lem:PauliDist}, and estimate the resulting sum by dividing the Pauli strings into several weight regimes.
We begin by decomposing the estimation bias in the Pauli basis.

\begin{lemma}[Estimation bias in terms of Pauli coefficients]\label{lem:BiasPauli}
For a given state $\rho$ and a stabilizer-state observable $O$, the estimation bias of the SPS protocol is
    \begin{equation}\label{Eq:BiasPauli}
   \text{Bias}(\widehat{o_f}) = D^{-1}|\sum_{P\in \mb{P}_n\setminus\mc{Z}_n}(\sigma_P-1)\tr(\rho P)\tr(OP)|\leq D^{-1}\sum_{P\in \mathrm{STAB}(O)\setminus\mc{Z}_n}|\sigma_P-1|.
\end{equation}
\end{lemma}
Here, the coefficient $\sigma_P$ is shown in \cref{Eq:sigmaP}. For a stabilizer-state observable $O=V\ket{\mb{0}}\!\bra{\mb{0}}V^{\dagger}$, we denote by $\mathrm{STAB}(O)$ the stabilizer group of $V\ket{\mb{0}}$.

\begin{proof}
    For a given state $\rho$ and a stabilizer state $O$, the estimation bias under the  measurement channel of the shallow phase shadow is
\begin{equation}\label{Eq:GeneralBias_ap}
    \begin{split}
       \text{Bias}(\widehat{o_f}) &= |\mbb{E}_{U\sim\mc{E}_{\text{sparse}}^{(\gamma)}}\tr(\widehat{\rho_f} O)- \tr(\rho_f O)|\\
    &=| \mbb{E}_{U\sim \mc{E}_{\text{sparse}}^{(\gamma)}}\sum_{\mb{b}} \tr(\rho\Phi_{U,\mb{b}})\tr[(D\Phi_{U,\mb{b}}-\id)O] - \tr(\rho_f O)|\\
    & = |\mbb{E}_{U\sim \mc{E}_{\text{sparse}}^{(\gamma)}}\sum_{\mb{b}} \tr(\rho\Phi_{U,\mb{b}})\tr[D\Phi_{U,\mb{b}}O] -\tr(O)- \tr(\rho_f O)|\\
    &=|D^2\tr(\rho\otimes O \text{ }\mb{M}_{\mc{E}_{\text{sparse}}^{(\gamma)}}^{(2)}) -\tr(O)- \tr(\rho_f O)|.\\
    \end{split}
\end{equation}

Next, we represent the estimation bias in the Pauli basis by introducing the Pauli form of the second moment function from \cref{Eq: P2P}. We note that for $P \in \mc{Z}_n$, the coefficient $\sigma_{\id_D} = D$, while $\sigma_P = 0$ for $P\in \mc{Z}_n/\{\id_D\}$.
\begin{equation}
    \begin{split}
        \text{Bias}(\widehat{o_f}) &= |D^2\tr(\rho\otimes O \text{ }\mb{M}_{\mc{E}_{\text{sparse}}^{(\gamma)}}^{(2)}) -\tr(O)-\tr(\rho_f O)|\\
        &= |D^{-1}\tr(\rho\otimes O \text{ }\sum_{P\in\mb{P}_n}\sigma_P P\otimes P )- \tr(O)-D^{-1}\tr(\rho\otimes O \text{ }\sum_{P\in\mb{P}_n/\mc{Z}_n} P\otimes P )|\\
        &= |D^{-1}\tr(\rho\otimes O \text{ }\sum_{P\in\mb{P}_n/\mc{Z}_n}\sigma_P P\otimes P )-D^{-1}\tr(\rho\otimes O \text{ }\sum_{P\in\mb{P}_n/\mc{Z}_n} P\otimes P )|\\
        &=D^{-1}|\tr[\rho\otimes O\text{ }\sum_{P\in\mb{P}_n/\mc{Z}_n}(\sigma_P-1) P\otimes P]|=D^{-1}|\sum_{P\in \mb{P}_n\setminus\mc{Z}_n}(\sigma_P-1)\tr(\rho P)\tr(OP)|\\
        &\leq D^{-1}\sum_{P\in \mb{P}_n\setminus\mc{Z}_n }|\sigma_P-1||\tr(OP)|\leq D^{-1}\sum_{P\in \text{STAB}(O)\setminus\mc{Z}_n }|\sigma_P-1|.
    \end{split}
\end{equation}
\end{proof}
 With $O$ being a stabilizer-state observable, we obtain that $|\tr(OP)|=1$ for Pauli operators $P$ belonging to the stabilizer group of $O$ (of size $D$), and $|\tr(OP)|=0$ for all other Pauli operators. Consequently, bounding the estimation bias reduces to controlling the sum of $|\sigma_P-1|$ over the $D$ stabilizer Paulis $P$ associated with $O$.
Before carrying out the proof of the estimation bias, we also collect three elementary inequalities that will be used later.

 \begin{fact}[Technical inequalities]
The following elementary inequalities will be used repeatedly in the proof of the theorem.
\begin{equation}\label{Eq:math1}
    C_n^k\leq 2^{nH(\frac{k}{n})}, 1\leq k\leq n,
\end{equation}
where $H(p) = -p\log_2p-(1-p)\log_2(1-p)$,
and
\begin{equation}\label{Eq:math2}
    F(n,k)\leq 1.5 C_n^k+1,  k\leq 0.2n
\end{equation}
as well as
\begin{equation}\label{Eq:math3}
    [1+n^{-\frac{2n_3\gamma}{n}}]^n \leq 2^n n^{-n_3\gamma\lambda},
\end{equation}
for any $\lambda>0$ such that $e^{-t}-1+\lambda t\leq 0$, with $t = \frac{2n_3\gamma}{n}\ln n$.
 \end{fact}
\begin{proof}
\cref{Eq:math1} is a standard bound on binomial coefficients using Shannon entropy.

For $k<n/2$, we use $\sum_{i=0}^k C_n^i
\le
\frac{n-k+1}{n-2k+1}C_n^k.$
For $k\le 0.2n$, the prefactor is at most $1.5$ for $n\ge 1$, which gives~\cref{Eq:math2}.

Finally, we prove \cref{Eq:math3} as follows.
Let $t = \frac{2n_3\gamma}{n}\ln n$. Using $\ln(1+x)\le x$ and the assumption $e^{-t}-1+\lambda t\le0$, we obtain
\begin{equation}
\ln\left(\frac{1+e^{-t}}{2}\right)
=
\ln\left(1+\frac{e^{-t}-1}{2}\right)
\le
\frac{e^{-t}-1}{2}
\le
-\frac{\lambda t}{2}.
\end{equation}
Multiplying by $n$ gives
\begin{equation}
n\ln\left(\frac{1+e^{-t}}{2}\right)
\le
-\frac{n\lambda t}{2}
=
-n_3\gamma\lambda\ln n.
\end{equation}
Exponentiating both sides yields
\begin{equation}
\left(\frac{1+n^{-\frac{2n_3\gamma}{n}}}{2}\right)^n
\le
n^{-n_3\gamma\lambda},
\end{equation}
which is equivalent to \eqref{Eq:math3}.
\end{proof}

\textbf{Final proof of the bias bound.}
We now combine \cref{lem:PauliDist} and \cref{lem:BiasPauli} to bound the estimation bias.
The proof proceeds by a case-by-case decomposition according to the $X$-weight
$n_3=\mathrm{wt}_X(P)$ of the Pauli operators appearing in \cref{Eq:BiasPauli}.
By grouping the stabilizer Paulis according to their $X$- and $Z$-weights,
$n_3=\mathrm{wt}_X(P)$ and $n_2=\mathrm{wt}_Z(P)$, \cref{Eq:BiasPauli} gives
\begin{equation}\label{Eq:BiasByWeights}
\mathrm{Bias}(\widehat{o_f})
\le
D^{-1}\sum_{n_3=1}^{n}
\sum_{n_2=0}^{n-n_3}
\sum_{\substack{
P\in \mathrm{STAB}(O)\setminus\mc Z_n\\
\mathrm{wt}_X(P)=n_3,\ \mathrm{wt}_Z(P)=n_2
}}
|\sigma_P-1| = D^{-1}\sum_{n_3=1}^{n}
\sum_{n_2=0}^{n-n_3} N_{X,Z}(n_3,n_2)|\sigma(n_2,n_3)-1|,
\end{equation}
where the notation $N_{X,Z}(n_3,n_2)$ is defined in~\cref{eq:Nxz} and can be bounded using~\cref{lem:PauliDist}. Here, $\sigma_P$ depends only on the Pauli-weight data $(n_1,n_2,n_3)$, not on the locations of the non-identity Pauli operators. Therefore, after fixing $n_2$ and $n_3$ with $n_1=n-n_2-n_3$, all Paulis in the corresponding class have the same coefficient, which we denote by $\sigma(n_2,n_3)$.

In the following, we split the contribution of estimation bias in~\cref{Eq:BiasByWeights} into the following four regimes:
\begin{align*}
\text{Case 1:}\quad &  n_3\gamma \le 0.0517\,\frac{n}{\ln n},\\
\text{Case 2:}\quad & 0.0517\,\frac{n}{\ln n} < n_3\gamma \le \frac{4n}{\ln n},\\
\text{Case 3:}\quad & n_3\gamma > \frac{4n}{\ln n}\quad \text{and}\quad n_3\le 0.2n,\\
\text{Case 4:}\quad & n_3>0.2n.
\end{align*}
The four cases reflect the trade-off between the size of the stabilizer support
and the decay of the coefficient $\sigma(n_2,n_3)$.
Case 1 is controlled mainly by counting, since only a limited number of small-$X$-weight stabilizer Paulis can appear.
Case 4 uses the fact that large $n_3$ forces $\sigma(n_2,n_3)$ to be close to its
ideal value $1$, so each individual term is strongly suppressed.
Cases 2 and 3 are intermediate regimes where the counting estimate and the decay
of $\sigma(n_2,n_3)$ together make the contribution negligible.
The constants are chosen for technical convenience and are not optimized.

Throughout the proof, we assume $\gamma>3$ and $n$ is sufficiently large so that
$2\gamma\ln n/n<1$. Besides, we denote
\begin{equation}\label{eq:qgamma}
    q(\gamma,n_3):=(1-\frac{2\gamma \ln n}{n})^{n_3}\in(0,1)
\end{equation}
throughout this proof and recall that \cref{lem:PauliDist} ensures the existence of a parameter $\zeta=\rank(C)\ge n/2$ for stabilizer-state observables $O$.

\textbf{ Case 1: $\mathbf{n_3\gamma\leq0.0517\frac{n}{\ln n}}$.}
Here, we focus on the corresponding contribution to the estimation bias in \cref{Eq:BiasPauli} $D^{-1}\sum_{n_3=1}^{0.0517\tfrac{n}{\gamma \ln n}}
\sum_{n_2=0}^{n-n_3} N_{X,Z}(n_3,n_2)|\sigma(n_2,n_3)-1|$. For any choice of $n_2$, the independent term is upper-bounded as follows.
\begin{equation}\label{Eq:Biascase1}
    \begin{split}
        &D^{-1}N_{X,Z}(n_3,n_2)|\sigma(n_2,n_3)-1|\leq
        D^{-1}N_{X,Z}(n_3,n_2)(\sigma(n_2,n_3)+1)\\
        &\leq
        D^{-1}N_{X}(n_3)
        +D^{-1}N_{X,Z}(n_3,n_2)\sigma(n_2,n_3)\qquad\text{(using $N_{X,Z}(n_3,n_2)\le N_X(n_3)$)}\\
        &\leq
        D^{-1}F(\zeta,n_3)2^{n-\zeta}
        +D^{-1}2^{n_3}[F(\zeta,n_3)-1]F(n-\zeta,n_2)\sigma(n_2,n_3)
        \qquad(\text{using the counting bounds in~\cref{lem:PauliDist}})\\
        &\leq
        D^{-1}\times 1.5 C_\zeta^{n_3}2^{\frac{n}{2}}
        +D^{-1}\times 1.5\,C_{\zeta}^{n_3}2^{n_3}F(n-\zeta,n_2)
        [1+(1-\tfrac{2\gamma \ln n}{n})^{n_3}]^{n_1}
        [1-(1-\tfrac{2\gamma \ln n}{n})^{n_3}]^{n_2}\\
        &\qquad \text{(using $\zeta\geq n/2$; see~\cref{Eq:math2})}\\
        &\leq2^{-n/2}\times 1.5 C_\zeta^{n_3}
        +D^{-1}\times 1.5\,C_{\zeta}^{n_3}2^{n_3}F(n-\zeta,n_2)
        [1+(1-\tfrac{2\gamma \ln n}{n})^{n_3}]^{n}
        [1-(1-\tfrac{2\gamma \ln n}{n})^{n_3}]^{n_2}\qquad \text{(using $n_1\le n$)}\\
        &\leq
        \negl(n)
        +D^{-1}\times 1.5\,C_{\zeta}^{n_3}2^{n_3}F(n-\zeta,n_2)
        [1+n^{-\frac{2n_3\gamma}{n}}]^{n}
        [1+(1-\tfrac{2\gamma \ln n}{n})^{n_3}]^{-n_2}
        [1-(1-\tfrac{2\gamma \ln n}{n})^{n_3}]^{n_2}\\
        &\qquad(C_\zeta^{n_3}=\mathrm{subexp}(n) \text{ since } n_3=O(\tfrac{n}{\ln n}); 1+x\leq e^x\text{ with } x = -\tfrac{2\gamma \ln n}{n})\\
        &\leq\negl(n)+1.5\,C_{\zeta}^{n_3}2^{n_3}n^{-0.95n_3\gamma}F(n-\zeta,n_2)
        [1+(1-\tfrac{2\gamma \ln n}{n})^{n_3}]^{-n_2}
        [1-(1-\tfrac{2\gamma \ln n}{n})^{n_3}]^{n_2}\\
        &\qquad(\text{by~\cref{Eq:math3} with $\lambda=0.95$, since $t=\frac{2n_3\gamma}{n}\ln n\le0.1034$ implies $e^{-t}-1+0.95t\le0$})\\
        &=\negl(n)
        +1.5\,C_{\zeta}^{n_3}2^{n_3}n^{-0.95n_3\gamma}\times F(n-\zeta,n_2)
        [1+q(\gamma,n_3)]^{-n_2}
        [1-q(\gamma,n_3)]^{n_2}.\\
    \end{split}
\end{equation}
It remains to bound the last three terms $ F(n-\zeta,n_2)
        [1+q(\gamma,n_3)]^{-n_2}
        [1-q(\gamma,n_3)]^{n_2}$ in the last line of \cref{Eq:Biascase1}, which are dependent on $n_2$ .
\begin{equation}\label{Eq:Biascase1_1}
\begin{split}
&F(n-\zeta,n_2)
\left[\frac{1-q(\gamma,n_3)}{1+q(\gamma,n_3)}\right]^{n_2} =
\sum_{i=0}^{n_2} C_{n-\zeta}^{i}
\left[\frac{1-q(\gamma,n_3)}{1+q(\gamma,n_3)}\right]^{n_2}\qquad(\text{see definition of }F(a,b)\text{ in~\cref{lem:PauliDist}})\\
&\le
\sum_{i=0}^{n_2} C_{n-\zeta}^{i}
\left[\frac{1-q(\gamma,n_3)}{1+q(\gamma,n_3)}\right]^{i}\le
\sum_{i=0}^{n-\zeta} C_{n-\zeta}^{i}
\left[\frac{1-q(\gamma,n_3)}{1+q(\gamma,n_3)}\right]^{i}=
\left[
1+\frac{1-q(\gamma,n_3)}{1+q(\gamma,n_3)}
\right]^{n-\zeta}
\le
\left[
1+\frac{1-q(\gamma,n_3)}{1+q(\gamma,n_3)}
\right]^{n/2}.
\end{split}
\end{equation}

We then note that $\frac{1-x}{1+x}\le 1-x^{0.55}$ for $0.8967\le x<1$. This inequality is applicable to $x=q(\gamma,n_3)$ in the present regime, since Bernoulli's inequality gives
$
q(\gamma,n_3)=\left(1-\frac{2\gamma\ln n}{n}\right)^{n_3}
\ge 1-\frac{2n_3\gamma\ln n}{n}
\ge 0.8967,
$
where the last inequality follows from the condition of this regenime, $n_3\gamma\le 0.0517\,n/\ln n$. In this way, \cref{Eq:Biascase1_1} can be further converted to that
\begin{equation}\label{Eq:Biascase1_2}
\begin{split}
&\left[
1+\frac{1-q(\gamma,n_3)}{1+q(\gamma,n_3)}
\right]^{n/2} \le
\left[2-q(\gamma,n_3)^{0.55}\right]^{n/2}=
\left[
2-\left(1-\frac{2\gamma\ln n}{n}\right)^{0.55n_3}
\right]^{n/2}\\
&\le
\left[
2-\left(1-\frac{1.1\gamma n_3\ln n}{n}\right)
\right]^{n/2}\qquad(\text{ }(1-x)^a\geq1-ax, \text{ when }a=0.55n_3>1, \text{ with }x=\tfrac{2\gamma \ln n}{n})\\
&=
\left[
1+\frac{1.1\gamma n_3\ln n}{n}
\right]^{n/2}\le
\exp\!\left(0.55\gamma n_3\ln n\right)
=
n^{0.55\gamma n_3}.
\end{split}
\end{equation}
Based on the result of~\cref{Eq:Biascase1_2}, \cref{Eq:Biascase1} now delivers
\begin{equation}\label{eq:Biasecase1_21}
    D^{-1}\sum_{\mathrm{wt}_X(P)=n_3,\mathrm{wt}_Z(P)=n_2}|\sigma(n_2,n_3)-1|\leq \text{negl}(n)+1.5C_{\zeta}^{n_3}2^{n_3}n^{-0.95\gamma n_3+0.55\gamma n_3}=\text{negl}(n)+1.5C_{\zeta}^{n_3}2^{n_3}n^{-0.4\gamma n_3}.
\end{equation}

Note that \cref{Eq:Biascase1_2} is satisfied only when $0.55n_3>1$, which does not include the subcase when $n_3=1$. To tackle this exceptional issue, we bound this subcase independently in a similar way. Recall that when $n_3=1$, the coefficient
\begin{equation}
    \sigma(n_2,n_3) = (2-\frac{2\gamma\ln n}{n})^{n_1} (\frac{2\gamma\ln n}{n})^{n_2}\leq (2-\frac{2\gamma\ln n}{n})^{n-n_2} (\frac{2\gamma\ln n}{n})^{n_2}\leq 2^n n^{-\gamma}(\frac{\gamma\ln n}{n-\gamma \ln n})^{n_2},
\end{equation}
where  we use $(1+x)\leq e^x$ in the last inequality, with $x = -\tfrac{2\gamma \ln n}{n}$. Then, the bias contribution when $n_3=1$ is bounded by
\begin{equation}
    \begin{split}
        &D^{-1}   N_{X,Z}(1,n_2)|\sigma(n_2,n_3)-1|\leq D^{-1}N_{X,Z}(1,n_2)(\sigma(n_2,n_3)+1)=D^{-1}N_{X,Z}(1,n_2)+D^{-1}N_{X,Z}(1,n_2)\sigma(n_2,n_3)\\
        &\leq \negl(n)+D^{-1}2^{1}[F(\zeta,1)-1]F(n-\zeta,n_2)\sigma(n_2,n_3)\qquad(\text{following the proof in~\cref{Eq:Biascase1}})\\
        &\leq \negl(n) + 2D^{-1} C_\zeta^1 F(n-\zeta,n_2) (\frac{\gamma\ln n}{n-\gamma\ln n})^{n_2} 2^n  n^{-\gamma}\leq \negl(n) + 2C_\zeta^1n^{-\gamma}\sum_{i=0}^{n_2}C_{n-\zeta}^{i} (\frac{\gamma\ln n}{n-\gamma \ln n})^{i}\\
        &\leq \negl(n) + 2C_\zeta^1n^{-\gamma}(1+\frac{\gamma\ln n}{n-\gamma \ln n})^{n-\zeta}\leq  \negl(n) + 2C_\zeta^1n^{-\gamma}(1+\frac{1.2\gamma\ln n}{n})^{n/2}\qquad (\text{when probability }p= \tfrac{\gamma \ln n}{n}\leq \tfrac{1}{6})\\
        &\leq  \negl(n) + 2C_\zeta^1n^{0.6\gamma-\gamma}=\negl(n) + 2C_\zeta^1n^{-0.4\gamma},
    \end{split}
\end{equation}
which is no more than~\cref{eq:Biasecase1_21} when $\gamma>3$. Thus, the subcase $n_3=1$ also follows the bound in \cref{eq:Biasecase1_21} by taking $n_3=1$ there.

Summing over all contributions (all possible $n_2$ and $n_3$) to the estimation bias in Case~1, we obtain from~\cref{eq:Biasecase1_21} that
\begin{equation}\label{Eq:Biascase1_3}
 \begin{split}
& D^{-1}
\sum_{1\le n_3\le  \frac{0.0517n}{\gamma\ln n}}
\sum_{n_2=0}^{n-n_3}
N_{X,Z}(n_3,n_2)\,|\sigma(n_2,n_3)-1|  \le
\sum_{n_3=1}^{\frac{0.0517n}{\gamma\ln n}}
\sum_{n_2=0}^{n-n_3}
\bigl[\negl(n)+1.5C_\zeta^{n_3}2^{n_3}n^{-0.4\gamma n_3}\bigr] \\
&\le
\negl(n)
+n\sum_{n_3=1}^{ \frac{0.0517n}{\gamma\ln n}}
1.5C_\zeta^{n_3}2^{n_3}n^{-0.4\gamma n_3} \le
1.5\,n\sum_{n_3=1}^{\zeta} C_\zeta^{n_3}
\bigl(2n^{-0.4\gamma}\bigr)^{n_3}
=
1.5\,n\Bigl[\bigl(1+2n^{-0.4\gamma}\bigr)^\zeta-1\Bigr] \le
4n^2n^{-0.4\gamma}.
 \end{split}
\end{equation}
where in the last step we have used $\zeta\le n$ and
$(1+x)^m-1\le \frac{4}{3}mx$ for sufficiently small $x=2n^{-0.4\gamma}$.

\textbf{Case 2: $0.0517\,\frac{n}{\ln n}< n_3\gamma\leq \frac{4n}{\ln n}$. }
Fix $n_3$ in this case and an arbitrary $n_2$. The contribution of this Pauli weight to~\cref{Eq:BiasByWeights} is bounded by
\begin{equation}\label{Eq:Biascase2}
    \begin{split}
        &D^{-1}N_{X,Z}(n_3,n_2)
        |\sigma(n_2,n_3)-1| \\
        &\leq
        D^{-1}N_{X,Z}(n_3,n_2)\bigl(\sigma(n_2,n_3)+1\bigr) \\
        &=
        D^{-1}N_{X,Z}(n_3,n_2)
        +
        D^{-1}N_{X,Z}(n_3,n_2)\sigma(n_2,n_3)\\
        &=
        \negl(n)
        +
        D^{-1}N_{X,Z}(n_3,n_2)\sigma(n_2,n_3)
        \qquad(\text{follow~\cref{Eq:Biascase1}})\\
        &\leq
        \negl(n)
        +D^{-1}2^{n_3}[F(\zeta,n_3)-1]F(n-\zeta,n_2)
        [1+q(\gamma,n_3)]^{n_1}[1-q(\gamma,n_3)]^{n_2}
        \qquad(\text{using~\cref{eq:Nxz}})\\
        &\leq
        \negl(n)
        +D^{-1}2^{n_3}F(\zeta,n_3)F(n-\zeta,n_2)
        [1+q(\gamma,n_3)]^{n}[1-q(\gamma,n_3)]^{n_2}.
    \end{split}
\end{equation}
Following the proof in \cref{Eq:Biascase1_1}, one arrives at $F(n-\zeta,n_2)[1-q(\gamma,n_3)]^{n_2}\leq[2-q(\gamma,n_3)]^{\frac{n}{2}}$. Then, \cref{Eq:Biascase2} gives
\begin{equation}\label{Eq:Biascase2_1}
    \begin{split}
        D^{-1}N_{X,Z}(n_3,n_2)|\sigma(n_2,n_3)-1|
        &\leq
        \negl(n)
        +D^{-1}F(\zeta,n_3)2^{n_3}
        [1+q(\gamma,n_3)]^n[2-q(\gamma,n_3)]^{n/2}\\
        &=
        \negl(n)
        +2^{n_3}F(\zeta,n_3)
        2^{n\left[\log_2(1+q(\gamma,n_3))
        +\frac{1}{2}\log_2(2-q(\gamma,n_3))-1\right]}.
    \end{split}
\end{equation}
It remains to show that the exponential term in~\cref{Eq:Biascase2_1}
decays exponentially in $n$. Define $f(x):=\log_2(1+x)+\frac{1}{2}\log_2(2-x)-1$ with $x=q(\gamma,n_3)$. Then the exponent in~\cref{Eq:Biascase2_1} is $n f(q(\gamma,n_3))$.
Therefore, it suffices to prove $f(q(\gamma,n_3))<0$ in this case.

One can verify that $f(x):=\log_2(1+x)+\frac12\log_2(2-x)-1<0$
whenever $x<e^{-0.1}$. It therefore suffices to show that
$q(\gamma,n_3)<e^{-0.1}$. In the present case,
$n_3\gamma>0.0517\,n/\ln n$, and hence $\tfrac{2\gamma n_3\ln n}{n}>2\times 0.0517=0.1034>0.1.$
Using $1-y\le e^{-y}$, we obtain $q(\gamma,n_3)
=
\left(1-\frac{2\gamma\ln n}{n}\right)^{n_3}
\le
\exp\!\left(-\frac{2\gamma n_3\ln n}{n}\right)
<e^{-0.1}.$
Thus $f(q(\gamma,n_3))<0$, so the exponential factor in the second term of~\cref{Eq:Biascase2_1} decays exponentially in $n$.

Moreover, in the present case $n_3=O(n/\ln n)$, the prefactor $2^{n_3}F(\zeta,n_3)$ is at most sub-exponential in $n$. Hence, the second term in \cref{Eq:Biascase2_1} is also negligible, and therefore the whole contribution in this case is bounded by $\negl(n)$.

\textbf{Case 3: $\mb{n_3\gamma> \frac{4n}{\ln n}}$ and $\mb{n_3\leq 0.2n}$.} Fix $n_3$ in this case and an arbitrary $n_2$. We first bound
\begin{equation}\label{Eq:Biascase3}
\begin{split}
D^{-1}N_{X,Z}(n_3,n_2)
&\le
D^{-1}N_X(n_3) \le
D^{-1}\tfrac{\zeta-n_3+1}{\zeta-2n_3+1}C_\zeta^{n_3}2^{n-\zeta}
\qquad
(\text{using~\cref{lem:PauliDist}  })\\
&\le
3\times2^{-n}2^{\zeta H(n_3/\zeta)}2^{n-\zeta} \qquad(\text{using $n_3\leq 0.2n,\zeta\ge 0.5n$; using~\cref{Eq:math1}})\\
&=
3\times2^{\zeta[H(n_3/\zeta)-1]} \le
3\times2^{n\frac{\zeta}{n}\left[H(0.2n/\zeta)-1\right]} .
\end{split}
\end{equation}
Suppose $g(x) = x[H(\frac{0.2}{x})-1]$ with $x = \tfrac{\zeta}{n}$, where $\frac{1}{2}\leq x< 1$. One can prove that $g(x)\leq\frac{1}{2}[H(0.4)-1]=-0.0145$. Utilizing this result, one arrives at
\begin{equation}\label{Eq:Biascase3-1}
    \begin{split}
        &D^{-1}N_{X,Z}(n_3,n_2)|\sigma(n_2,n_3)-1|
        \leq
        D^{-1}N_{X,Z}(n_3,n_2)\bigl(1+\sigma(n_2,n_3)\bigr)\\
        &\leq
        3\cdot 2^{-0.0145n}
        \left(
        1+
        [1+q(\gamma,n_3)]^{n_1}[1-q(\gamma,n_3)]^{n_2}
        \right)\leq
        3\cdot 2^{-0.0145n}
        \left(1+[1+q(\gamma,n_3)]^{n}\right).
    \end{split}
\end{equation}

In this case, $n_3\gamma> \frac{4n}{\ln n}$, which leads to $q(\gamma,n_3) := (1-\frac{2\gamma\ln n}{n})^{n_3}\leq\exp(-\frac{2\gamma n_3\ln n}{n})\leq e^{-8}$. Hence, the corresponding share of the estimation bias in \cref{Eq:BiasPauli} is also $\text{negl}(n)$, as $2^{0.0145}>1+e^{-8}$, and therefore the whole contribution in this case is negligible.

\textbf{Case 4: $\mb{n_3> 0.2n}$.} The corresponding component of the estimation bias in \cref{Eq:BiasByWeights} is upper bounded by
\begin{equation}\label{Eq:Biascase4}
\begin{split}
&D^{-1}\sum_{n_3>0.2n}\sum_{n_2=0}^{n-n_3}N_{X,Z}(n_3,n_2)
|\sigma(n_2,n_3)-1| \leq
\max_{n_3>0.2n}|\sigma(n_2,n_3)-1|
\qquad
\left(|\mathrm{STAB}(O)|=D\right)\\
&\leq
\max_{n_3>0.2n}
\left(
[1+q(\gamma,n_3)]^n-1
\right)\leq
\left[1+n^{-0.4\gamma}\right]^n-1
\qquad
\left(
q(\gamma,n_3)\le e^{-2\gamma n_3\ln n/n}\le n^{-0.4\gamma}
\right)\\
&\leq
2n\times n^{-0.4\gamma}.
\end{split}
\end{equation}

Combining all the four regimes, Case 1 contributes at most
$4n^2n^{-0.4\gamma}$, Case 4 contributes at most
$2n\times n^{-0.4\gamma}$, while Cases 2 and 3 are negligible.
As a result, combining the above four cases yields the desired upper bound on the estimation bias, and hence completes the proof of the bias part of \cref{th:PshadowMain}.

\subsection{Estimation variance}\label{Ap:MainVar}
Having bounded the estimation bias, our focus now turns to the second part of \cref{th:PshadowMain}, i.e., the estimation variance. The variance analysis naturally splits into two conceptual steps. First, the variance is expressed in a form (\cref{Eq:MajorVariance}) related to the third-order moment function.
Second, the components of the moment function are analyzed by classifying its contributions according to the permutation structure of $S_3$. A combinatorial counting argument then shows that the total contribution of the suppressed terms is at most sub-exponential in $n$.

The detailed proof begins as follows. As we denote ${\Phi}_{U,\mb{b}}\coloneq U^{\dagger}\ket{\mb{b}}\bra{\mb{b}}U$, the variance of $\widehat{o_f} = \tr(O \widehat{\rho_f})$ is bounded by its third moment
 \begin{equation}
\begin{split}
\text{Var}(\widehat{o_f}) &\leq\mbb{E}_{\{U,\mathbf{b}\}} [\widehat{o_f}^2] = \mbb{E}_{\{U,\mathbf{b}\}} \tr(O\widehat{\rho_f})^2\\
 &=\mbb{E}_{U\sim \mc{E}_{{\text{sparse}}}^{(\gamma)}}\sum_{\mb{b}}\tr(\rho {\Phi}_{U,\mb{b}})\tr(O\widehat{\rho_f})^2\\
&=\mbb{E}_{U\sim \mc{E}_{{\text{sparse}}}^{(\gamma)}}\sum_{\mb{b}}\tr[(\rho\otimes O_f\otimes O_f) \text{ }({\Phi}_{U,\mb{b}}\otimes \widehat{\rho_f}\otimes\widehat{\rho_f})].\\
\end{split}
\label{Eq:invrobustvariance0}
\end{equation}
Recall that $\widehat{\rho_f} = D{\Phi}_{U,\mb{b}}-\id_D$, we arrive at
\begin{equation}\label{Eq:varM}
    \begin{split}
            \text{Var}(\widehat{o_f})&\leq\tr[D^2\mbb{E}_{U\sim \mc{E}_{{\text{sparse}}}^{(\gamma)}}\sum_{\mb{b}}{\Phi}_{U,\mb{b}}^{\otimes 3}\text{ }(\rho\otimes O_f\otimes O_f)]\\
            &=\tr[D^3\mb{M}_{\mc{E}_{\text{sparse}}^{(\gamma)}}^{(3)}\text{ }(\rho\otimes O_f\otimes O_f)],
        \end{split}
\end{equation}
where $\mb{M}_{\mc{E}_{\text{sparse}}^{(\gamma)}}^{(3)}\coloneq D^{-1}\mbb{E}_{U\sim \mc{E}_{{\text{sparse}}}^{(\gamma)}}\sum_{\mb{b}}{\Phi}_{U,\mb{b}}^{\otimes 3}$.
According to \cref{thm: sparse 3-moment},
\begin{equation}
    \mb{M}_{\mc{E}_{\text{sparse}}^{(\gamma)}}^{(3)}=D^{-3}\sum_{\mb{x+w+z=y+s+t}}(1-\frac{2\gamma\ln n}{n})^{N_{\mb{x,w,z,y,s,t}}}\ket{\mb{x,w,z}}\bra{\mb{y,s,t}},
\end{equation}
where
\begin{equation}\label{Eq:coef}
    N_{\mb{x,w,z,y,s,t}}\coloneq \sum_{i<j}\mb{1}\{x_ix_j+w_iw_j+z_iz_j\neq y_iy_j+s_is_j+t_it_j\pmod{2}\}.
\end{equation}
In order to bound the value of \cref{Eq:varM}, we divide the summation into two terms according to whether $N_{\mb{x,w,z,y,s,t}}=0$ or not.

\begin{equation}    \label{Eq:MajorVariance}
    \begin{split}
        &\tr[D^3\mathbf{M}_{\mc{E}_{\text{sparse}}^{(\gamma)}}^{(3)}\text{ }(\rho \otimes O_f\otimes O_f)] \\
        &=\mathop {\underbrace{\tr[(\rho \otimes O_f\otimes O_f)\text{ }\sum_{\mb{x}+\mb{w}+\mb{z}=\mb{y}+\mb{s}+\mb{t},{N_{\mb{x,w,z,y,s,t}}}=0 }\ket{\mb{x,w,z}}\bra{\mb{y,s,t}}  ]}} \limits_{\circled{1}}\\
        &+\mathop {\underbrace{\tr[(\rho \otimes O_f\otimes O_f)\text{ }\sum_{\mb{x}+\mb{w}+\mb{z}=\mb{y}+\mb{s}+\mb{t},{N_{\mb{x,w,z,y,s,t}}}\neq0 }(1-\frac{2\gamma\ln n}{n})^{N_{\mb{x,w,z,y,s,t}}} \ket{\mb{x,w,z}}\bra{\mb{y,s,t}}  ]}} \limits_{\circled{2}}.\\
    \end{split}
\end{equation}

As shown in Ref.~\cite{zhang2025robust}, $\sum_{\mb{x}+\mb{w}+\mb{z}=\mb{y}+\mb{s}+\mb{t},{N_{\mb{x,w,z,y,s,t}}}=0 }\ket{\mb{x,w,z}}\bra{\mb{y,s,t}}=\bigcup_{\pi \in S_3} V_n(\pi)$. Hence, the term $\circled{1}$ can be bounded as
$\circled{1}\leq 3\|O_f\|_2^2$. Then, our main focus is on bounding the term $\circled{2}$, which must be sub-exponential for the theorem to hold.

Next, we turn to the more challenging contribution $\circled{2}$ in~\cref{Eq:MajorVariance}, which we aim to bound to be sub-exponential. For the computational-basis elements $\ket{\mb{x,w,z}}\bra{\mb{y,s,t}}$ in $\circled{2}$, the condition $N_{\mathbf{x},\mathbf{w},\mathbf{z},\mathbf{y},\mathbf{s},\mathbf{t}}\neq0$ implies that the corresponding matrix unit does not belong to any global $n$-qubit permutation operators $V_n(\pi)$. However, the restriction $\mb{x}+\mb{w}+\mb{z}=\mb{y}+\mb{s}+\mb{t}$ still guarantees that the element belongs to a tensor product of single-qubit permutations $\ket{\mb{x,w,z}}\bra{\mb{y,s,t}}\in [\bigcup_{\pi\in S_3}V_1(\pi)]^{\otimes n}$. Here, the phrase `belongs to' follows the definition in~\cref{App:UaA}.

In this way, it is convenient to decompose the single-qubit operator space according to the permutation group $S_3=\{\pi_{()},\pi_{(23)},\pi_{(12)},\pi_{(13)},\pi_{(123)},\pi_{(132)}\}$. At the single-qubit level, $V_1(\pi)$ denotes the action of a permutation $\pi\in S_3$ on three copies of one qubit.
The subspace fixed by all such permutations is $\Delta_3\coloneq\ket{0,0,0}\bra{0,0,0}+\ket{1,1,1}\bra{1,1,1}=\bigcap_{\pi \in S_3}V_1(\pi)$. To organize the contribution $\circled{2}$, we refine the local three-copy operator space according to its permutation profile. Specifically, we denote $B(\pi_1,\pi_2) \coloneq V_1(\pi_1)\cap V_1(\pi_2)-\Delta_3$, where $\pi_1,\pi_2$ are two different permutations. The purpose of introducing $B(\pi_1,\pi_2)$ is to group those matrix elements that share the same decay exponent $N_{\mb{x,w,z,y,s,t}}$.

We now summarize several elementary properties of the sectors $B(\pi_1,\pi_2)$ as follows. First, by direct inspection, $B(\pi_1,\pi_2)\neq\emptyset$ only when one permutation is odd and the other is even. We therefore introduce
\begin{equation}
S_{3,\mathrm{odd}}:=\{\pi_{(23)},\pi_{(12)},\pi_{(13)}\},
\qquad S_{3,\mathrm{even}}:=\{\pi_{()},\pi_{(123)},\pi_{(132)}\}.
\end{equation}
With the convention that the first label is odd and the second label is even, there are exactly nine nontrivial sectors $B(\pi_1,\pi_2)$, where we set $\pi_1\in S_{3,\mathrm{odd}}$ and $\pi_2\in S_{3,\mathrm{even}}$.
Second, each of these nine sectors $B(\pi_1,\pi_2)$ contains exactly two rank-one computational-basis operators of the form
$\ket{a,b,c}\bra{a',b',c'}$. Finally, $B(\pi_1,\pi_2)$ and $B(\pi'_1,\pi'_2)$ have no intersections when either $\pi_1\neq \pi'_1$ or $\pi_2\neq \pi'_2$.

In the single-qubit setting, the constraint
$a+b+c=a'+b'+c'$ selects $20$ rank-one computational-basis operators.
These operators consist of the two diagonal elements in $\Delta_3$, together with the $18$ off-diagonal elements grouped into the nine mutually disjoint sectors $B(\pi_1,\pi_2)$.

Hence, one can conclude that given six $n$-bit binary strings $\mb{x},\mb{w},\mb{z},\mb{y},\mb{s},\mb{t}$,
\begin{equation}\label{eq:3ordersum}
\sum_{\mb{x}+\mb{w}+\mb{z}=\mb{y}+\mb{s}+\mb{t}}\ket{\mb{x,w,z}}\bra{\mb{y,s,t}}= [\Delta_3+\sum_{\pi_1\in S_{3,\mathrm{odd}},\pi_2\in S_{3,\mathrm{even}}}B(\pi_1,\pi_2)]^{\otimes n},
\end{equation}

and specifically
\begin{equation}\label{Eq:PermSum}
    \sum_{\pi_2\in S_{3,\mathrm{even}}} B(\pi_1,\pi_2)+\Delta_3 = V_1(\pi_1) \text{    and   }\sum_{\pi_1\in S_{3,\mathrm{odd}}} B(\pi_1,\pi_2)+\Delta_3 = V_1(\pi_2).
\end{equation}

We refer to $\mathfrak I
:=
\Bigl(
I^{(0)},
\{I^{(\pi_1,\pi_2)}\}_{\pi_1\in S_{3,\mathrm{odd}},\,\pi_2\in S_{3,\mathrm{even}}}
\Bigr)$ as the position information, where these subsets form a partition of $[n]$, namely $I^{(0)}
\cup
\Bigl(
\mathop{\bigcup}_{\pi_1\in S_{3,\mathrm{odd}},\,\pi_2\in S_{3,\mathrm{even}}}
I^{(\pi_1,\pi_2)}
\Bigr)
=
[n]$.

Therefore, every rank-one operator
$
\ket{\mb{x,w,z}}\bra{\mb{y,s,t}}
$
belongs to (see the definition in~\cref{App:UaA}) a unique tensor-product block of the form
\begin{equation}\label{eq:def-B-frakI-m3}
B_{\mathfrak I}
:=
\Delta_3^{\otimes I^{(0)}}
\otimes
\bigotimes_{\pi_1\in S_{3,\mathrm{odd}},\,\pi_2\in S_{3,\mathrm{even}}}
B(\pi_1,\pi_2)^{\otimes I^{(\pi_1,\pi_2)}},
\end{equation}
Here $I^{(0)}$ records the positions whose local factor lies in $\Delta_3$, and
$I^{(\pi_1,\pi_2)}$ records the positions whose local factor lies in
$B(\pi_1,\pi_2)$.

Next, we calculate the value of the coefficient $N_{\mb{x,w,z,y,s,t}}$ using the concept of $B(\pi_1,\pi_2)$. We can check that given two single-qubit operators $\ket{x_1,w_1,z_1}\bra{y_1,s_1,t_1}\in B(\pi_1,\pi_2)$ and $\ket{x_2,w_2,z_2}\bra{y_2,s_2,t_2}\in B(\pi'_1,\pi'_2)$, the condition
\begin{equation}\label{eq:condition}
    x_1x_2+w_1w_2+z_1z_2\neq y_1y_2+s_1s_2+t_1t_2
\end{equation}
is satisfied if and only if $\pi_1\neq \pi'_1$ and $\pi_2\neq\pi'_2$. Consequently, the condition \cref{eq:condition} is entirely determined by the permutation elements $\pi_1, \pi_2, \pi'_1, \pi'_2$.  For every $n$-qubit operator $\ket{\mb{x,w,z}}\bra{\mb{y,s,t}}\in B_{\mathfrak I}$, the corresponding coefficient defined in \cref{Eq:coef} becomes
\begin{equation}\label{Eq:CoefN}
    N_{\mb{x,w,z,y,s,t}} = \frac{1}{2}\sum_{\pi_1\in S_{3,\mathrm{odd}},\pi_2\in S_{3,\mathrm{even}}}|I^{(\pi_1,\pi_2)}|[\sum_{\pi'_1\in S_{3,\mathrm{odd}}/\{\pi_1\},\pi'_2\in S_{3,\mathrm{even}}/\{\pi_2\}}|I^{(\pi'_1,\pi'_2)}|],
\end{equation}
which depends only on the cardinalities of the sets $I^{(0)}$ and $I^{(\pi_1,\pi_2)}$, denoted as $|I^{(0)}|$ and $|I^{(\pi_1,\pi_2)}|$, respectively. Here, we add the coefficient $\frac12$ as every term $|I^{(\pi_1,\pi_2)}|\times |I^{(\pi'_1,\pi'_2)}|$ occurs twice. In particular, the value of $N_{\mb{x,w,z,y,s,t}}$ depends only on the cardinalities $|I^{(0)}|$ and $|I^{(\pi_1,\pi_2)}|$. Hence, we group together all $n$-qubit operators that share the same
element-count information. To this end, we denote
\begin{equation}
    \mathbf{I}
= \bigl[ |I^{(0)}|,\, |I^{(\pi_1,\pi_2)}| \text{ for every } (\pi_1,\pi_2) \bigr].
\end{equation}
For a fixed element-count vector $\mathbf I$, let $\mathcal P(\mathbf I)$ denote the set of position information $\mathfrak I$ whose cardinalities agree with $\mathbf I$. Then, every $n$-qubit operator $\ket{\mb{x,w,z}}\bra{\mb{y,s,t}}$ in the tensor-product block $B_\mathfrak{I}$ with $\mathfrak I\in\mc{P}(\mb{I})$ shares the same coefficient ${N}_{\mb{x,w,z,y,s,t}}$ according to~\cref{Eq:CoefN}, which we denote as ${N}_\mb{I}$.

With this notation, the full three-copy support in~\cref{eq:3ordersum} can be regrouped as
\begin{equation}\label{eq:regroup}
\sum_{\mb{x}+\mb{w}+\mb{z}=\mb{y}+\mb{s}+\mb{t}}\ket{\mb{x,w,z}}\bra{\mb{y,s,t}}= [\Delta_3+\sum_{\pi_1\in S_{3,\mathrm{odd}},\pi_2\in S_{3,\mathrm{even}}}B(\pi_1,\pi_2)]^{\otimes n}=\sum_{\mb{I}}
\sum_{\mathfrak I\in\mathcal P(\mathbf I)} B_{\mathfrak I}.
\end{equation}
The total number of admissible element-count configurations $\mathbf I$ is at most
$C_{n+9}^9$ since there are nine possible nontrivial pairs $(\pi_1,\pi_2)$ together with the residual sector $I^{(0)}$, which grows only polynomially with $n$. We also define \(\mc C_{\mathbf I}\) as the set of all rank-one matrix units belonging to tensor-product blocks with the same element-count information $\mc C_{\mathbf I}
:=
\cup_{\mathfrak I\in\mathcal P(\mathbf I)} B_{\mathfrak I}$.

The contribution $\circled{2}$ in~\cref{Eq:MajorVariance} can
now be decomposed according to the element-count
information $\mb{I}$ as
\begin{equation}\label{Eq:Var2Split}
\begin{split}
\circled{2}
&=\tr\!\Bigl[
(\rho \otimes O_f \otimes O_f)
\sum_{\substack{\mb{x}+\mb{w}+\mb{z}=\mb{y}+\mb{s}+\mb{t}\\
N_{\mb{x,w,z,y,s,t}}\neq 0}}
\Bigl(1-\frac{2\gamma\ln n}{n}\Bigr)^{N_{\mb{x,w,z,y,s,t}}}
\ket{\mb{x,w,z}}\bra{\mb{y,s,t}}
\Bigr] \\
&=\sum_{\mathbf I:\,N_{\mathbf I}>0}
\Bigl(1-\frac{2\gamma\ln n}{n}\Bigr)^{N_{\mathbf I}}
\tr\!\Bigl[
(\rho \otimes O_f \otimes O_f)
\sum_{{\mathfrak{I}}\in \mc{P}(\mathbf I)}B_{\mathfrak{I}}
\Bigr].\qquad(\text{see~\cref{eq:regroup}})\\
&=\sum_{\mathbf I:\,N_{\mathbf I}>0}
\Bigl(1-\frac{2\gamma\ln n}{n}\Bigr)^{N_{\mathbf I}}
\tr\!\Bigl[
(\rho \otimes O_f \otimes O_f)
\sum_{\ket{\mb{x,w,z}}\bra{\mb{y,s,t}}\in \mc{C}_{\mb{I}}}\ket{\mb{x,w,z}}\bra{\mb{y,s,t}}
\Bigr].\\
\end{split}
\end{equation}
where we have used~\cref{eq:regroup} and the fact that all elements in the same class \(\mc C_{\mathbf I}\) share the same exponent $N_{\mathbf I}$.
Since the number of element-count information $\mb{I}$ is at most $C_{n+9}^9=\mathrm{poly}(n)$, to show that $\circled{2}$ is sub-exponential, it suffices to bound
each individual contribution associated with a fixed $\mathbf I$. We proceed by estimating these terms separately.

Here, we observe that the matrix elements of $\rho\otimes O_f\otimes O_f$ may have arbitrary signs or phases. However, we can replace $\rho$ and $O_f$ by entrywise nonnegative operators $\tilde{\rho}$ and $\tilde{O}$.  After this replacement, if two Boolean matrices $A$ and $B$ satisfy $A\subseteq B$ entrywise, then we have
\begin{equation}\label{eq:AcompareB}
\tr\!\left[(\tilde{\rho}\otimes\tilde O\otimes\tilde O)A\right]
\le
\tr\!\left[(\tilde{\rho}\otimes\tilde O\otimes\tilde O)B\right].
\end{equation}

This monotonicity under support inclusion is the key reason for introducing the tilde operation below.

We first show how to replace $\rho$ and $O_f$ by entrywise nonnegative operators $\tilde{\rho}$ and $\tilde{O}$.
Given an arbitrary Hermitian operator $H\coloneq \sum_{i,j}{h_{i,j}}\ket{i}\bra{j}$, we denote $\tilde{H}\coloneq \sum_{i,j}\sqrt{h_{i,i}h_{j,j}}\ket{i}\bra{j}$, and $\tilde{H}$ is positive semidefinite, where there exists a vector $\ket{\psi_H} = \sum_{i}\sqrt{h_{i,i}}\ket{i}$, such that $\tilde{H}=\ket{\psi_{H}}\bra{\psi_{H}}$. If $H\succeq 0$, then we have $\tilde{h}_{i,j}=\sqrt{h_{i,i}h_{j,j}}\ge |h_{i,j}|$. In this way, for a certain element-count information $\mb{I}$, the corresponding component in~\cref{Eq:Var2Split} is bounded by

\begin{equation}\label{Eq:VarSep}
    \begin{split}
        &\tr[(\rho \otimes O_f\otimes O_f)\text{ }\sum_{\ket{\mb{x,w,z}}\bra{\mb{y,s,t}}\in \mc{C}_{\mb{I}}}\ket{\mb{x,w,z}}\bra{\mb{y,s,t}}]= \sum_{\ket{\mb{x,w,z}}\bra{\mb{y,s,t}}\in \mc{C}_{\mb{I}}} \bra{\mb{y}}\rho\ket{\mb{x}} \bra{\mb{s}}O_f\ket{\mb{w}} \bra{\mb{t}}O_f\ket{\mb{z}}\\
        &\leq \sum_{\ket{\mb{x,w,z}}\bra{\mb{y,s,t}}\in \mc{C}_{\mb{I}}} |\rho_{\mb{x,y}}|\times|[O_f]_{\mb{w,s}}|\times|[O_f]_{\mb{z,t}}|
        \leq  \sum_{\ket{\mb{x,w,z}}\bra{\mb{y,s,t}}\in \mc{C}_{\mb{I}}} |\rho_{\mb{x,y}}|\times|O_{\mb{w,s}}|\times|O_{\mb{z,t}}|\\
        &\leq  \sum_{\ket{\mb{x,w,z}}\bra{\mb{y,s,t}}\in \mc{C}_{\mb{I}}} \tilde{\rho}_{\mb{x,y}}\times\tilde{O}_{\mb{w,s}}\times\tilde{O}_{\mb{z,t}} = \tr[(\tilde{\rho} \otimes \tilde{O}\otimes \tilde{O})\text{ }\sum_{\ket{\mb{x,w,z}}\bra{\mb{y,s,t}}\in \mc{C}_{\mb{I}}}\ket{\mb{x,w,z}}\bra{\mb{y,s,t}}]\\
        &=\tr[(\tilde{\rho} \otimes \tilde{O}\otimes \tilde{O})\text{ }\sum_{{\mathfrak{I}}\in \mc{P}(\mathbf I)}B_{\mathfrak{I}}].
\end{split}
\end{equation}
\begin{figure}
    \centering
    \includegraphics[width=\linewidth]{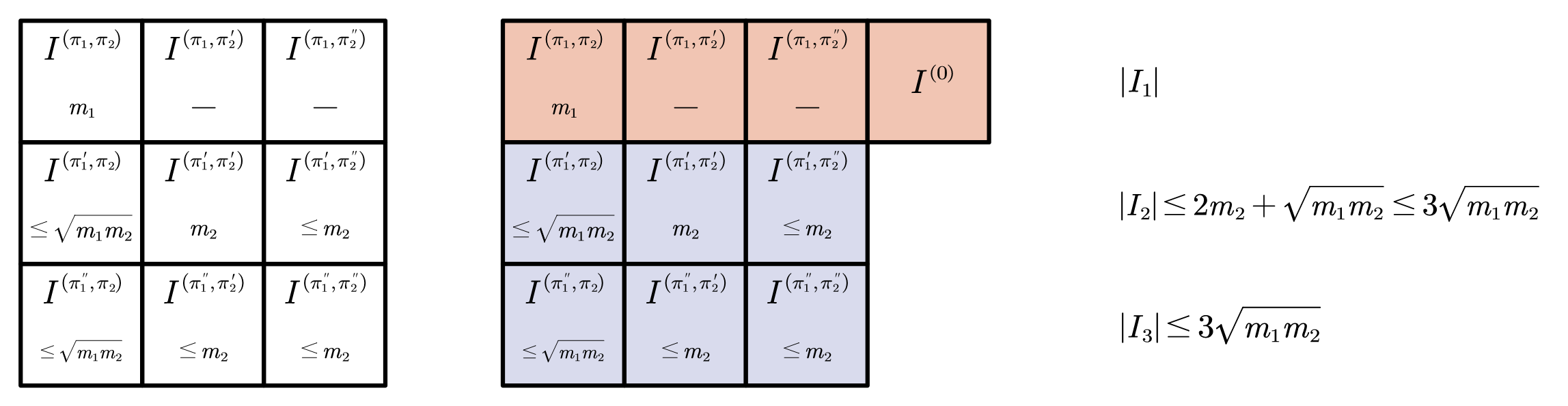}
\caption{
Schematic organization of the permutation sectors used in the variance bound.
The sectors are indexed by an odd label $\pi_1$ and an even label $\pi_2$.
We choose two dominant sectors with different row and column labels, of sizes
$m_1$ and $m_2$, whose product is maximal.
This maximality bounds the relevant non-dominant sectors by $\sqrt{m_1m_2}$, so that after grouping sites by odd labels the two smaller groups have sizes at most $3\sqrt{m_1m_2}$.
This yields the multiplicity estimate
$(C_n^{3\sqrt{m_1m_2}})^2$
appearing in the variance bound.
}
    \label{fig:placeholder}
\end{figure}
Next, we utilize~\cref{eq:AcompareB} to bound~\cref{Eq:VarSep} as follows.
As is shown in~\cref{fig:placeholder}, we denote
$ m_1 \coloneq |I^{(\pi_1,\pi_2)}| $ and
$ m_2 \coloneq |I^{(\pi'_1,\pi'_2)}| $,
where
\begin{equation}
    (\pi_1,\pi_2),(\pi'_1,\pi'_2)
= \arg\max_{\pi_1 \neq \pi'_1,\ \pi_2 \neq \pi'_2}
|I^{(\pi_1,\pi_2)}|\cdot|I^{(\pi'_1,\pi'_2)}|.
\end{equation}
The product
$|I^{(\pi_1,\pi_2)}||I^{(\pi'_1,\pi'_2)}|=m_1m_2$ is maximal among all pairs with different row and column labels. Without loss of generality, we set $m_1\ge m_2$.

Next, we control the sectors $(\pi_1,\pi_2)$ lying in the same row or the same column as the dominant sector
$(\pi_1,\pi_2)$.  Consider the two collections
\begin{equation}
\mathfrak R:=\{(\pi_1,\pi_2^*):\pi_2^*\neq \pi_2\},
\qquad
\mathfrak C:=\{(\pi_1^*,\pi_2):\pi_1^*\neq \pi_1\}.
\end{equation}
Choose the larger maximum between these two collections. Without loss of generality, assume that it is attained in the row collection, and denote it by
\begin{equation}\label{eq:maxoutside}
(\pi_1,\pi'_2)
\in
\arg\max_{(\pi_1^*,\pi_2^*)\in \mathfrak R\cup\mathfrak C}
|I^{(\pi_1^*,\pi_2^*)}|.
\end{equation}
The case where the maximum is attained in the column collection is identical
after exchanging the roles of the row and column labels, and therefore does
not affect the generality of the argument.
For any $\pi_1^*\neq \pi_1$, the two sectors
$(\pi_1,\pi'_2)$ and $(\pi_1^*,\pi_2)$ have different row and column labels. By the maximality of $m_1m_2$, we have $|I^{(\pi_1,\pi'_2)}|\,|I^{(\pi_1^*,\pi_2)}|
\le m_1m_2$.
Besides, by the choice of $(\pi_1,\pi'_2)$ as the largest sector in $\mathfrak R\cup\mathfrak C$ in~\cref{eq:maxoutside}, we obtain $|I^{(\pi_1,\pi'_2)}|\ge |I^{(\pi_1^*,\pi_2)}|.$ Combining the two inequalities gives
\begin{equation}
|I^{(\pi_1^*,\pi_2)}|\le \sqrt{m_1m_2},
\qquad \forall\,\pi_1^*\neq \pi_1 .
\end{equation}
An analogous argument applies to the sectors in the same row or column after exchanging the two labels.

Indeed, for any fixed position information $\mathfrak I$, each qubit position belongs
either to $\Delta_3$ or to one of the sectors
$B(\pi_1,\pi_2)$, with $\pi_1\in S_{3,\mathrm{odd}}$ and
$\pi_2\in S_{3,\mathrm{even}}$. We now group the qubit positions according to their
odd label. In particular, for every $\pi_1$, define $I_1
:=
I^{(0)}
\cup
\bigcup_{\pi_2^*\in S_{3,\mathrm{even}}}
I^{(\pi_1,\pi_2^*)},$ where the qubit positions in $I^{(0)}$ are assigned here only. Then, we use the remaining two odd permutations $\pi_1', \pi''_1$ to define $I_2\coloneq\bigcup_{\pi_2^*\in S_{3,\mathrm{even}}}I^{(\pi_1',\pi_2^*)}$ and $I_3\coloneq\bigcup_{\pi_2^*\in S_{3,\mathrm{even}}}I^{(\pi_1'',\pi_2^*)}$.

For qubit positions in $I^{(\pi_1^*,\pi_2^*)}$, we have $B(\pi_1^*,\pi_2^*)\subseteq V_1(\pi_1^*).$
For qubit positions in $I^{(0)}$, we have $\Delta_3\subseteq V_1(\pi_1)$. Therefore,
\begin{equation}
B_{\mathfrak I}
\subseteq
V_1(\pi_1)^{\otimes I_1}\otimes V_1(\pi_1')^{\otimes I_2}\otimes V_1(\pi_1'')^{\otimes I_3},
\end{equation}
 Taking the union over all position information $\mathfrak I$ compatible
with the same element-count vector $\mathbf I$ gives
\begin{equation}
\sum_{{\mathfrak{I}}\in \mc{P}(\mathbf I)}B_{\mathfrak{I}}
\subseteq \bigcup _{I_1+I_2+I_3=[n]}V_1(\pi_1)^{\otimes I_1}\otimes V_1(\pi_1')^{\otimes I_2}\otimes V_1(\pi_1'')^{\otimes I_3}.
\end{equation}
As illustrated in
Fig.~\ref{fig:placeholder}, the maximality of $m_1m_2$ implies
\begin{equation}
    |I_2|\le 2m_2+\sqrt{m_1m_2}\le 3\sqrt{m_1m_2},
\qquad
|I_3|\le 3\sqrt{m_1m_2}.
\end{equation}
In this way, \cref{Eq:VarSep} can be bounded using~\cref{eq:AcompareB} as
\begin{equation}\label{Eq:Holder}
    \begin{split}
        &\tr[(\tilde{\rho} \otimes \tilde{O}\otimes \tilde{O})\text{ }\sum_{{\mathfrak{I}}\in \mc{P}(\mathbf I)}B_{\mathfrak{I}}]\leq \tr[(\tilde{\rho} \otimes \tilde{O}\otimes \tilde{O})\text{ }\bigcup _{I_1+I_2+I_3=[n]}V_1(\pi_1)^{\otimes I_1}\otimes V_1(\pi_1')^{\otimes I_2}\otimes V_1(\pi_1'')^{\otimes I_3}]\\
        &\leq \sum_{I_1+I_2+I_3 = [n]}\tr[(\tilde{\rho} \otimes \tilde{O}\otimes \tilde{O})\text{ }V_1(\pi_1)^{\otimes I_1}\otimes V_1(\pi_1')^{\otimes I_2}\otimes V_1(\pi_1'')^{\otimes I_3}]\\
        & \leq\sum_{I_1+I_2+I_3 = [n]} \|V_1(\pi_1)^{\otimes I_1}\otimes V_1(\pi_1')^{\otimes I_2}\otimes V_1(\pi_1'')^{\otimes I_3}\|_\infty \|\tilde{\rho} \otimes \tilde{O}\otimes \tilde{O}\|_1,
    \end{split}
\end{equation}
due to the Hölder's inequality. In addition, we have $\|V_1(\pi_1)^{\otimes I_1}\otimes V_1(\pi_1')^{\otimes I_2}\otimes V_1(\pi_1'')^{\otimes I_3}\|_\infty=\|\tilde{\rho} \otimes \tilde{O}\otimes \tilde{O}\|_1=1$, where $\tilde{O}$ is also a quantum state.  Applying these norm bounds in \cref{Eq:Holder}, we obtain
\begin{equation}\label{Eq:VarFinal1}
    \tr[(\tilde{\rho} \otimes \tilde{O}\otimes \tilde{O})\text{ }\sum_{{\mathfrak{I}}\in \mc{P}(\mathbf I)}B_{\mathfrak{I}}]\leq \sum_{I_1+I_2+I_3 = [n]} \mb{1} = C_n^{|I_2|} C_{n-|I_2|}^{|I_3|} .
\end{equation}
Combining~\cref{Eq:VarFinal1} with the decay factor in~\cref{Eq:Var2Split}, and using~\cref{Eq:CoefN}, we have
\begin{equation}
N_{\mathbf I}\ge m_1m_2,
\qquad
\left(1-\frac{2\gamma\ln n}{n}\right)^{N_{\mathbf I}}
\le\exp\!\left[-\frac{2\gamma\ln n}{n}m_1m_2\right].
\end{equation}
Hence the contribution of a fixed class $\mathbf I$ in~\cref{Eq:Var2Split} is bounded by the function $F(\mb{I})\coloneq C_n^{|I_2|}C_{n-|I_2|}^{|I_3|}
\exp\!\left[-\frac{2\gamma\ln n}{n}m_1m_2\right]$. If \begin{equation}
t:=\sqrt{m_1m_2}=\Theta(n),
\end{equation}
then the exponential factor is
$\exp[-\Omega(n\ln n)]$, which dominates the combinatorial prefactor.
Thus, this regime is negligible. It remains to consider $t=o(n)$. In order to show that the estimation variance is sub-exponential, we need to prove that $\ln F(\mb{I})/n = o(1)$.
\begin{equation}\label{Eq:VarSubExp}
\begin{split}
       \frac{\ln F(\mb{I})}{n}\le \frac{\ln\{(C_n^{3\sqrt{m_1m_2}})^2\exp[-\frac{2\gamma\ln n}{n}m_1m_2]\}}{n} &= \frac{\ln\{(C_n^{3t})^2\exp[-\frac{2\gamma\ln n}{n}t^2]\}}{n}<\frac{6t}{n}\ln\frac{ne}{3t}-2\gamma\ln n (\frac{t}{n})^2 +O(\frac{1}{n}),\\
\end{split}
\end{equation}
due to $C_n^k\leq (\frac{ne}{k})^k$. Let $x=\frac{t}{n}\in(0,1)$, we obtain $y = 6x\ln\frac{e}{3}-6x\ln x-2\gamma\ln n x^2$, whose derivative $y' = -6\ln 3-6\ln x-4\gamma x\ln n$. The function $y$ reaches the maximum when  $y'=0$, i.e., $6\ln 3+6\ln x^*-4\gamma x^*\ln n=0$, in which case $y_{\max} = 3x^*(2-\ln3x^*)$. As $y'|_{x=\frac{\ln\ln n}{\ln n}}=-6\ln3-6\ln\ln\ln n+6\ln\ln n-4\gamma \ln\ln n<0$ when $\gamma\geq1.5$, we set $x^* \leq \frac{\ln\ln n}{\ln n} $, and $y_{\max} < \frac{3\ln\ln n}{\ln n}(2-\ln 3-\ln\ln\ln n+\ln\ln n)$ converges to zero. In this way,
we obtain that \cref{Eq:Var2Split} is sub-exponential.

Since the term $\circled{1}$ is constant and term $\circled{2}$ is sub-exponential, the total variance is sub-exponential, which completes the final proof of \cref{th:PshadowMain}.

Having proven the single-shot variance bound for the off-diagonal estimator, we now comment on the total estimation variance of the SPS protocol in~\cref{algo:PhaseShadow}. Recall that the second-look principle uses two measurement tracks. Therefore, if the total number of phase-shadow samples is $N_f$ and the total number of computational-basis samples is $N_d$, then the selected track contains $N_f/2$ phase-shadow samples and $N_d/2$ computational-basis samples.

For the diagonal part, direct computational-basis measurements give
the single-shot estimator $\widehat{o_d}=\tr(O\ket{\mathbf b}\bra{\mathbf b}),$
which is unbiased for $\tr(O\rho_d)$. Moreover,
\begin{equation}
    \mathrm{Var}(\widehat{o_d})
    \le
    \mathbb E_{\mathbf b}\widehat{o_d}^{\,2}
    =
    \sum_{\mathbf b}\bra{\mathbf b}\rho\ket{\mathbf b}
    \bra{\mathbf b}O\ket{\mathbf b}^{2}
    \le
    \|O_d\|_\infty^2
    \le
    \|O_d\|_2^2 .
\end{equation}
Let $\widehat{o_f}$ denote the phase-shadow estimator selected by the second-look rule, and let
$\widehat{o}=\widehat{o_f}+\widehat{o_d}$. Since the two data sets are sampled independently, the variance satisfies
\begin{equation}
    \mathrm{Var}(\widehat{o})
    =
    \frac{2}{N_f}\mathrm{Var}(\widehat{o_f})
    +
    \frac{2}{N_d}\mathrm{Var}(\widehat{o_d}) .
\end{equation}
Using the single-shot variance bound proved above for the off-diagonal estimator, we obtain
\begin{equation}\label{}
    \mathrm{Var}(\widehat{o})
    \le
    \frac{2}{N_f}\mathrm{subexp}(n)
    +
    \frac{2}{N_d}\|O_d\|_2^2 .
\end{equation}
Here $\mathrm{subexp}(n)$ denotes the single-shot variance bound for the off-diagonal phase-shadow estimator. Moreover, when $O$ is a stabilizer-state observable, $\|O_d\|_2^2$ is at most a constant, since $O$ is a rank-one projector and
\begin{equation}
    \|O_d\|_2^2 \le \|O\|_2^2=\tr(O^2)=1 .
\end{equation}
Therefore, after averaging over $N_f$ phase-shadow samples and $N_d$ computational-basis samples, the full SPS estimator has sub-exponential variance in $n$. This is the variance control needed for the median-of-means procedure in \cref{algo:PhaseShadow}. Applying the standard median-of-means concentration bound separately to the off-diagonal and diagonal data sets yields the high-probability estimator used in the main text.

\section{Details on the numerical simulations}

\subsection{Numerical verification of bias scaling with $\gamma$}\label{app:bias_with_gamma}

In this appendix, we provide numerical evidence showing that the bias often decays much faster than predicted by the theoretical bound in~\cref{th:PshadowMain}. Although the proof requires $\gamma>3$, the numerical results show that accurate estimation is already achieved for smaller values of $\gamma$.
To verify this, we have performed numerical simulations on a system of $n=50$ qubits. We have evaluated the estimation bias for three distinct classes of quantum states:
\begin{itemize}
    \item The state vector $|\psi\rangle = |+\rangle^{\otimes n/2} \otimes |0\rangle^{\otimes n/2}$, representing a scenario where the second-look principle yields no additional benefit.
    \item The state vector $|\psi\rangle = |+\rangle^{\otimes n-2} \otimes |0\rangle^{\otimes 2}$, corresponding to the example discussed in~\cref{fig:hadamard_trick} in the main text.
    \item The cluster star state vector $|\psi\rangle = \prod_{j=1}^{n-1}CZ_{0,j}|+\rangle^{\otimes n}$, a graph state with high entanglement connectivity, selected as a complementary case featuring significant entanglement.
\end{itemize}

The numerical results are presented in \cref{fig:bias_scaling}, with the theoretical upper bound derived in~\cref{Eq:Bias} indicated by the gray dash-dotted line. Several key observations can be made. First, the simulations confirm that the analytical bias bound is rigorous but conservative. To quantify this, we extracted the empirical scaling exponents (normalized by $\ln n$) from the decay slopes, yielding coefficients of approximately $-0.88$, $-1.19$, and $-1.57$ for the three tested states. These values significantly outperform the theoretical scaling coefficient of $-0.4$, confirming that the empirically observed bias is orders of magnitude lower than the theoretical prediction. Second, although the theorem formally requires $\gamma > 3$ to guarantee convergence, the bias decreases rapidly (exponentially) even for significantly smaller values of $\gamma$. For instance, at $\gamma=1.5$, the bias for the cluster star state and the $|+\rangle^{\otimes n-2}|0\rangle^{\otimes 2}$ state has already dropped below $10^{-2}$ and $10^{-1}$ respectively, rendering the approximation highly accurate for practical purposes.

\begin{figure}[htbp]
    \centering
    \includegraphics[width=0.5\linewidth]{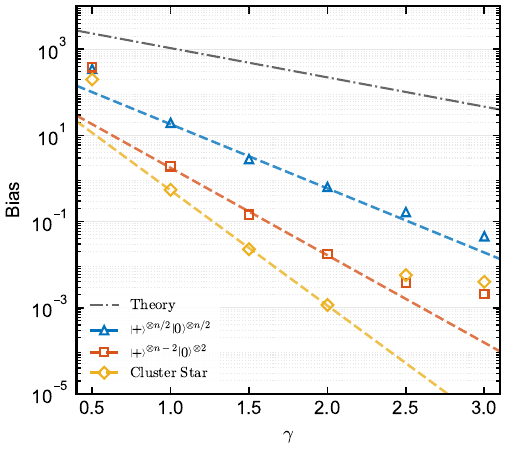}
    \caption{\label{fig:bias_scaling} Scaling of the estimation bias with respect to the parameter $\gamma$ for $n=50$ qubits. The performance is evaluated on three different states: the half-filled state $|+\rangle^{\otimes n/2}|0\rangle^{\otimes n/2}$ (blue triangles), the state $|+\rangle^{\otimes n-2}|0\rangle^{\otimes 2}$ (orange squares), and the cluster star state (yellow diamonds). Dashed lines represent linear fits to the logarithmic data, indicating exponential decay. Note that the linear fitting is restricted to the intermediate regime.
    For $\gamma=0.5$ and $\gamma>2$, the bias of the latter two states is comparable to or smaller than the statistical fluctuations, and is therefore excluded from the linear fit. The gray dash-dotted line denotes the theoretical upper bound.}
\end{figure}

\subsection{Standardization of circuit resources in~\cref{fig:guarentee}}
\label{app:gate_count_mapping}

To facilitate an equitable comparison across distinct protocols, we map the protocol-specific control parameters to a unified metric: the total two-qubit gate count ($N_g$). Given the disparate circuit architectures employed and considering their shallow implementations facilitated by auxiliary qubits, a simple metric like circuit depth fails to accurately capture the actual resource cost. Consequently, $N_g$ serves as a robust proxy for the total circuit volume and the associated accumulated noise cost.

The mapping logic for each protocol is as follows:
\begin{itemize}
    \item One-dimensional  brickwork circuits~\cite{bertoni2024shallow}: The control parameter is the circuit depth $d$. On average, a random two-qubit Clifford gate decomposes into 1.5 two-qubit gates (CNOTs). Since one layer consists of about $n/2$ disjoint pairs, increasing $d$ by 1 adds approximately $0.75n$ two-qubit gates. Consequently, the granularity of resource adjustment coarsens as the system size $n$ increases.
    \item Gluing random circuits~\cite{schuster2024random}: The control parameter is the patch size, denoted as $2\xi$. As $\xi$ increases, the number of two-qubit gates required to construct the localized unitary increases, leading to a higher $N_g$. To strictly adhere to the theoretical bounds derived in the original proposal, the patch size $2\xi$ must be a divisor of the system size $n$. Deviating from this constraint results in uncovered qubits at the boundaries, hindering a rigorous performance comparison under the theorem's description.
    \item SPS (sparse Clifford-IQP): The control parameter is $\gamma$. This parameter allows for continuous tuning of the gate probability $p = \gamma \ln n / n$, enabling precise control over the expected total number of two-qubit gates, unlike the discrete steps required by the other protocols.
\end{itemize}

Table~\ref{tab:gate_count_comparison} summarizes the correspondence between these parameters and $N_g$. Due to the stringent divisibility constraints of the gluing constructions, we established its valid gate counts as the baseline targets for designing the experiments for the 1D brickwork circuits and SPS protocols. This approach ensures comparisons are made at similar resource levels.

\begin{table*}[htbp]
    \centering
    \small
    \renewcommand{\arraystretch}{1.3}
    \caption{Alignment of control parameters and gate counts.
    Correspondence between protocol-specific parameters (SPS $\gamma$, Gluing patch size $2\xi$, and 1D brickwork layer depth $d$) and the total gate count ($N_g$). The columns are aligned by approximate gate count levels to demonstrate the resource comparability across schemes.
    The SPS protocol includes additional data points in the shallow circuit regime ($N_g < 180$) where the other protocols are not applicable or comparable.}
    \label{tab:gate_count_comparison}

    \begin{tabularx}{\textwidth}{l C C C C C C C C C C C}
        \hline\hline
        \textbf{Metric} & \multicolumn{11}{c}{\textbf{Gate count expectation (approximate)}} \\
        \hline

        \textit{\textbf{SPS}} & & & & & & & & & & & \\
        \hspace{1em} Param $\gamma$ & 0.4 & 0.8 & 1.2 & 1.6 & 2.0 & 3.0 & 4.0 & 5.0 & 5.7 & 7.2 & -- \\
        \hspace{1em} Gates $N_g$    & \textbf{36} & \textbf{73} & \textbf{109} & \textbf{146} & \textbf{182} & \textbf{275} & \textbf{362} & \textbf{455} & \textbf{519} & \textbf{657} & -- \\
        \hline

        \textit{\textbf{Gluing random circuits}} & & & & & & & & & & & \\
        \hspace{1em} Param $2\xi$      & -- & -- & -- & -- & 4 & 6 & 8 & -- & 12 & 16 & 24 \\
        \hspace{1em} Gates $N_g$    & -- & -- & -- & -- & \textbf{182} & \textbf{275} & \textbf{362} & -- & \textbf{519} & \textbf{657} & \textbf{882} \\
        \hline

        \textit{\textbf{1D brickwork circuits}} & & & & & & & & & & & \\
        \hspace{1em} Param $d$      & -- & -- & -- & -- & 5 & 8 & 10 & -- & 15 & 19 & 25 \\
        \hspace{1em} Gates $N_g$    & -- & -- & -- & -- & \textbf{178} & \textbf{282} & \textbf{356} & -- & \textbf{533} & \textbf{675} & \textbf{884} \\
        \hline\hline
    \end{tabularx}
\end{table*}

\section{Details on the exactly unbiased shallow shadow}
\subsection{Proof of~\cref{th:PshadowUnbiased}}\label{Ap:ProofUnbiased}
Here, we prove  \cref{th:PshadowUnbiased} by showing that $\widehat{\rho_f}_{\text{sparse}}$ is indeed an unbiased estimator of $\rho_f$. Utilizing the result of \cref{Eq: P2P}, the proof is derived as
    \begin{equation}
        \begin{split}
            \mathbb{E}_{\{U,\mathbf{b}\}} \widehat{\rho_f}_{\text{sparse}} &= \sum_{\mb{b}}\mathbb{E}_{U\sim \mc{E}_\text{sparse}^{(\gamma)}} \Pr(\mb{b}|
{U})\sum_{P \in \mathbf{P}_n/\mathcal{Z}_n} \sigma_P^{-1} \, \tr(\Phi_{U,\mathbf{b}} P) \, P\\
&=D\sum_{P \in \mathbf{P}_n/\mathcal{Z}_n} \sigma_P^{-1} \, \tr( {\mathbf{M}}_{\mc{E}_\text{sparse}^{(\gamma)}}^{(2)}\text{ }\rho\otimes P) \, P\\
&=D^{-2}\sum_{P \in \mathbf{P}_n/\mathcal{Z}_n} \sigma_P^{-1} \, \tr( \sum_{P'\in\mb{P}_n}\sigma_{P'} P'\rho\otimes P'P) \, P\\
&=D^{-1}\sum_{P \in \mathbf{P}_n/\mathcal{Z}_n} \sigma_P^{-1} \,\sigma_{P} \tr(  P\rho) \, P=\rho_f,\\
        \end{split}
\end{equation}
where the last equality follows from the Pauli expansion of the off-diagonal component $\rho_f$.
\subsection{Post-processing cost on the unbiased shallow shadow}
\label{app:efficiency}

In the unbiased shadow estimation protocol (\cref{sec:unbiased}), the cost of classical post-processing is a key metric for assessing practical scalability. For the stabilizer-fidelity tasks considered here, the estimator can be evaluated directly on the relevant Pauli support rather than by constructing the full operator $\widehat{\rho_f}_{\mathrm{sparse}}$. This makes the post-processing efficient for the stabilizer-state observables studied in this work. To verify this, we performed benchmark tests on systems of varying sizes using a standard consumer-grade CPU (Intel Core i9-13900HX, RAM 16GB ).

We have measured the average post-processing time per shot for three representative quantum states across qubit numbers $n$ ranging from 15 to 65. The states evaluated were: the half-filled state $|+\rangle^{\otimes n/2}|0\rangle^{\otimes n/2}$, the full superposition state vector $|+\rangle^{\otimes n}$, and the cluster star state. The results are presented in \cref{fig:time_scaling}. Several key characteristics can be observed:
\begin{enumerate}
    \item Negligible time overhead: For all tested quantum states, the post-processing time for a single shot remains within the order of $10^{-2}$ seconds, even as the system size reaches $n=65$. Compared to the duration typically required for experimental data acquisition (shot collection), this classical computational overhead is virtually negligible.
    \item Favorable scalability: The post-processing time exhibits a moderate polynomial growth with respect to the number of qubits $n$ (approximately scaling as $O(n^2)$). This low-complexity scaling behavior indicates that the protocol can be effectively scaled to systems comprising hundreds of qubits without encountering classical computational bottlenecks.
    \item Robustness to state structure: Although slight differences in computation time exist among quantum states, the overall order of magnitude remains consistent. This suggests that the algorithmic efficiency is largely insensitive to the specific form of the target state.
\end{enumerate}

\begin{figure}[h]
    \centering
    \includegraphics[width=0.5\linewidth]{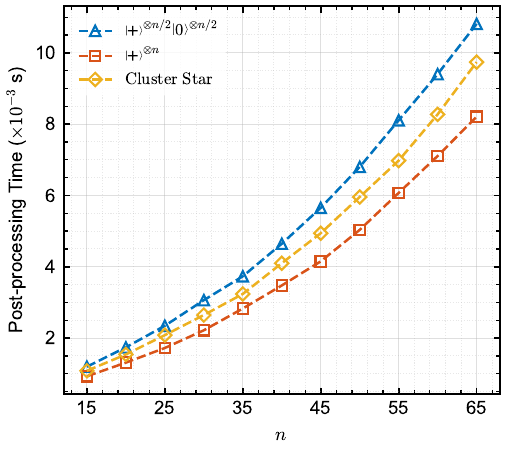}
    \caption{\label{fig:time_scaling} Classical post-processing time as a function of qubit number $n$. The plot shows the computational time required for three different quantum states: $|+\rangle^{\otimes n/2}|0\rangle^{\otimes n/2}$ (blue triangles), $|+\rangle^{\otimes n}$ (orange squares), and the cluster star state (yellow diamonds). Time is measured in milliseconds. The results demonstrate that even for a system size of $n=65$, the computation time is controlled within 12 ms, highlighting the high efficiency of the method.}
\end{figure}

\end{appendix}

\end{document}